\documentclass[11pt,a4paper,reqno]{amsproc}

\usepackage{datetime}
\usepackage[pdfborder={0 0 0}]{hyperref}
\usepackage{upref,cases}
  \hypersetup{
    citebordercolor=.1 .6 1,
    colorlinks = false,
    pdftitle = {The exponent in the orthogonality catastrophe for Fermi gases},
    pdfauthor = {Martin Gebert, Heinrich K\"uttler, Peter M\"uller, Peter Otte}
}

\usepackage[T1]{fontenc}
\usepackage[utf8]{inputenc}

\usepackage[final]{graphicx}

\IfFileExists{mmymtpro2.sty}{\usepackage[subscriptcorrection]{mymtpro2}

  \def\cap{\capprod}
  \def\bigcup{\bigcupprod}
  
  \def\bigcupdisjoint{\mathop{\kern10pt\raisebox{4pt}{$\cdot$}\kern-12pt\bigcup}\limits}
}%
	{\usepackage{fourier,times}            
}

\numberwithin{equation}{section}
\makeatletter

\newtheoremstyle{ttheorem}%
       {1.8ex\@plus1ex}                
       {2.1ex\@plus1ex\@minus.5ex}      
       {\itshape}           
       {0pt}                   
       {\bfseries}          
       {.}                  
       {.5em}               
       {}                

\newtheoremstyle{ddefinition}%
       {1.8ex\@plus1ex}                
       {2.1ex\@plus1ex\@minus.5ex}      
       {}           
       {0pt}                   
       {\bfseries}           
       {.}                  
       {.5em}               
       {}                

\newtheoremstyle{rremark}%
       {1.8ex\@plus1ex}                
       {2.1ex\@plus1ex\@minus.5ex}      
       {\normalfont}        
       {0pt}                   
       {\bfseries}           
       {.}                  
       {.5em}               
       {}                   

\theoremstyle{ttheorem}
\newtheorem{theorem}{Theorem}[section]
\newtheorem{lemma}[theorem]{Lemma}
\newtheorem{proposition}[theorem]{Proposition}
\newtheorem{corollary}[theorem]{Corollary}

\theoremstyle{ddefinition}
\newtheorem{definition}[theorem]{Definition}

\theoremstyle{rremark}
\newtheorem{remark}[theorem]{Remark}
\newtheorem{myremarks}[theorem]{Remarks}
\newtheorem{myexamples}[theorem]{Examples}

\newenvironment{remarks}{\begin{myremarks}\begin{nummer}}%
    {\end{nummer}\end{myremarks}}
    {\end{nummer}\end{myexamples}}

\newcounter{numcount}
\newcommand{\labelnummer}{(\roman{numcount})}%

\makeatletter
\providecommand{\showkeyslabelformat}[1]{\relax}        
\let\mysaveformat\showkeyslabelformat                   %
\def\myformat#1{\raisebox{-1.5ex}{\mysaveformat{#1}}}   %

\newenvironment{nummer}%
  {\let\curlabelspeicher\@currentlabel%
    \begin{list}{\textup{\labelnummer}}%
      {\usecounter{numcount}\leftmargin0pt%
        \topsep0.5ex\partopsep2ex\parsep0pt\itemsep0ex\@plus1\p@%
        \labelwidth2.5em\itemindent3.5em\labelsep1em%
      }%
    \let\saveitem\item%
    \def\item{\saveitem%
      \def\@currentlabel{\curlabelspeicher\kern.1em\labelnummer}}%
    \let\savelabel\label%
    \def\label##1{{\ifnum\thenumcount=1\let\showkeyslabelformat\myformat\fi\savelabel{##1}}%
										{\def\@currentlabel{\labelnummer}%
									 	\let\showkeyslabelformat\@gobble
									 	\savelabel{##1item}%
										}%
	   							}%
  }{\end{list}}%

  {\let\curlabelspeicher\@currentlabel%
    \begin{list}{\textup{\labelnummer}}%
      {\usecounter{numcount}\leftmargin0pt%
        \topsep0.5ex\partopsep2ex\parsep0pt\itemsep0ex\@plus1\p@%
        \labelwidth2.5em\itemindent0em\labelsep1em%
        \leftmargin2.5em}%
    \let\saveitem\item%
    \def\item{\saveitem%
      \def\@currentlabel{\curlabelspeicher\kern.1em\labelnummer}}%
    \let\savelabel\label%
    \def\label##1{{\ifnum\thenumcount=1\let\showkeyslabelformat\myformat\fi\savelabel{##1}}%
										{\def\@currentlabel{\labelnummer}%
									 	\let\showkeyslabelformat\@gobble
									 	\savelabel{##1item}%
										}%
    							}%
  }{\end{list}}%

\def\section{\@startsection{section}{1}%
  \z@{1.3\linespacing\@plus\linespacing}{.5\linespacing}%
  {\normalfont\bfseries\centering}}

\def\subsection{\@startsection{subsection}{2}%
  \z@{.8\linespacing\@plus.5\linespacing}{-1em}%
  {\normalfont\bfseries}}

\def\nlsubsection{\@startsection{subsection}{2}%
  \z@{.8\linespacing\@plus.5\linespacing}{.1ex}%
  {\normalfont\bfseries}}

\let\@afterindenttrue\@afterindentfalse%

\renewenvironment{proof}[1][\proofname]{\par \normalfont
  \topsep6\p@\@plus6\p@ \trivlist 
  \item[\hskip\labelsep\scshape
    #1\@addpunct{.}]\ignorespaces
}{%
  \qed\endtrivlist
}

\def\ps@firstpage{\ps@plain
  \def\@oddfoot{\normalfont\scriptsize \hfil\thepage\hfil
     \global\topskip\normaltopskip}%
  \let\@evenfoot\@oddfoot
  \def\@oddhead{
    \begin{minipage}{\textwidth}
      \normalfont\scriptsize
      \emph{\insertfirsthead}
    \end{minipage}}
  \let\@evenhead\@oddhead 
}

\def\insertfirsthead{}
\def\firsthead#1{\def\insertfirsthead{#1}}

\def\@cite#1#2{{%
 \m@th\upshape\mdseries[{#1}{\if@tempswa, #2\fi}]}}
\newcommand{\I}{\mathcal{I}}
\newcommand{\F}{\mathcal{F}}
\newcommand{\Nn}{\mathcal{N}}
\newcommand{\M}{\mathcal{M}}
\let\longHungarianUmlaut\H
\renewcommand{\H}{\mathcal{H}}

\newcommand{\SO}{\mathcal{S}}

\newcommand{\C}{\mathbb{C}}
\newcommand{\N}{\mathbb{N}}

\newcommand{\R}{\mathbb{R}}

\newcommand{\Z}{\mathbb{Z}}

\renewcommand{\le}{\leqslant}
\renewcommand{\ge}{\geqslant}

\DeclareMathOperator{\tr}{tr}

\providecommand{\bigcupdisjoint}{\mathop{\kern7pt\raisebox{6pt}{$\cdot$}\kern-9.5pt\bigcup}\limits}

\let\<\langle
\let\>\rangle

\providecommand{\abs}[1]{\lvert#1\rvert}
\providecommand{\norm}[1]{\lVert#1\rVert}

\providecommand{\bigabs}[1]{\bigl\lvert#1\bigr\rvert}
\providecommand{\Bigabs}[1]{\Bigl\lvert#1\Bigr\rvert}
\providecommand{\biggabs}[1]{\biggl\lvert#1\biggr\rvert}

\providecommand{\bignorm}[1]{\bigl\lVert#1\bigr\rVert}

\providecommand{\parens}[1]{(#1)}
\providecommand{\bigparens}[1]{\bigl(#1\bigr)}
\providecommand{\Bigparens}[1]{\Bigl(#1\Bigr)}
\providecommand{\biggparens}[1]{\biggl(#1\biggr)}

\providecommand{\bigbraces}[1]{\bigl\{#1\bigr\}}
\providecommand{\Bigbraces}[1]{\Bigl\{#1\Bigr\}}
\providecommand{\biggbraces}[1]{\biggl\{#1\biggr\}}

\providecommand{\bigangles}[1]{\bigl\<#1\bigr\>}

\let\set\braces
\let\bigset\bigbraces

\newcommand{\ac}{\mathrm{ac}}
\newcommand{\loc}{\mathrm{loc}}
\newcommand{\HS}{\mathrm{HS}}

\providecommand{\from}{\colon}
\providecommand{\st}{:\,}

\newcommand{\upto}{\uparrow}
\newcommand{\downto}{\downarrow}

\newcommand{\Oh}{\mathrm{O}}
\newcommand{\oh}{\mathrm{o}}

\newcommand{\where}[1]{
  \mathchoice{\qquad(#1)}%
             {\,\ (#1)}
             {\,(#1)}
             {\,(#1)}}

\newcommand{\dotid}{\,\boldsymbol\cdot\,}           

\DeclareMathOperator{\Borel}{Borel}      

\newcommand{\1}{1}

\newcommand{\upd}{\mathrm{d}}
\renewcommand{\d}{\upd}   
\newcommand{\dx}{\d x}
\newcommand{\dt}{\d t}
\newcommand{\dy}{\d y}

\newcommand{\dr}{\d r}
\newcommand{\du}{\d u}
\newcommand{\dS}{\d S}
\newcommand{\dbyd}[1]{\frac{\upd}{\upd{#1}}}

\newcommand{\spec}{\sigma}               

\DeclareMathOperator{\spt}{supp}          
\let\textdef\textit
\DeclareMathOperator{\ran}{ran}          
\newcommand{\hairspace}{\kern .04167em}

\let\sectionsymbol\S
\renewcommand{\S}{\mathcal{S}}

\newcommand{\Beta}{\mathrm{B}}

\DeclareMathOperator{\TextRe}{Re}
\DeclareMathOperator{\TextIm}{Im}
\renewcommand{\Re}{\TextRe}
\renewcommand{\Im}{\TextIm}

\newcommand{\Sphere}{\mathbb{S}}

\def\clap#1{\hbox to 0pt{\hss#1\hss}}

\def\bra{\makeatletter\@ifstar\@bra\@@bra}
\def\@bra#1{\hairspace #1\>}
\def\@@bra#1{\lvert\@bra{#1}}
\def\ket{\makeatletter\@ifstar\@ket\@@ket}
\def\@ket#1{\<#1\hairspace}
\def\@@ket#1{\@ket{#1}\rvert}

\def\author@andify{%
  \nxandlist {\unskip ,\penalty-1 \space\ignorespaces}%
    {\unskip {} \@@and~}%
    {\unskip \penalty-2 \space \@@and~}%
}
\renewcommand{\andify}{%
  \nxandlist{\unskip, }{\unskip{} \@@and~}{\unskip{} \@@and~}}

\makeatother
\begin{document}

\firsthead{To appear in J.\ Spectr.\ Theory} 

\title[The exponent in the orthogonality catastrophe]{The exponent in the orthogonality catastrophe\\ for Fermi gases}

\author[M.\ Gebert]{Martin Gebert}
\author[H.\ K\"uttler]{Heinrich K\"uttler}
\author[P.\ M\"uller]{Peter M\"uller}
\author[P.\ Otte]{Peter Otte}

\address[M.\ Gebert]{Mathematisches Institut,
  Ludwig-Maximilians-Universit\"at M\"unchen,
  Theresienstra\ss{e} 39,
  80333 M\"unchen, Germany}
\curraddr{Department of Mathematics, King's College London, Strand, London, WC2R 2LS, UK}
\email{martin.gebert@kcl.ac.uk}

\address[H.\ K\"uttler, P.\ M\"uller]{Mathematisches Institut,
  Ludwig-Maximilians-Uni\-ver\-si\-t\"at M\"unchen,
  Theresienstra\ss{e} 39,
  80333 M\"unchen, Germany}

\email{kuettler@math.lmu.de}
\email{mueller@lmu.de}

\address[P.~Otte]{Fernuniversit\"at in Hagen,
  Fachbereich Mathematik,
  LG Angewandte Stochastik,
  58084 Hagen, Germany}

\email{peter.otte@rub.de}

\thanks{Work supported by SFB/TR 12 of the German Research Council (DFG)}

\begin{abstract}
  We quantify the asymptotic vanishing of the ground-state overlap of two
  non-interacting Fermi gases in $d$-dimensional Euclidean space
  in the thermodynamic limit.
  Given two one-particle Schr\"odinger operators in finite-volume which differ by a compactly
  supported bounded potential, we prove a power-law upper bound on the ground-state
  overlap of the corresponding non-interacting $N$-particle systems.
  We interpret the decay exponent~$\gamma$ in terms of scattering
  theory and find $\gamma = \pi^{-2}\|\arcsin|T_E/2|\|_{\HS}^2$,
  where $T_E$ is the transition matrix at the Fermi energy~$E$. This
  exponent reduces to the one predicted by Anderson
  [Phys.\ Rev.\ \textbf{164}, 352--359 (1967)] for the exact asymptotics
  in the special case of a repulsive point-like perturbation.
\end{abstract}

\dedicatory{Dedicated to Hajo Leschke on the occasion of his 70$^{\text{th}}$ birthday}

\maketitle 

%
\section{Introduction}

We consider two quantum systems, each consisting of $N$
non-interacting Fermions in a box of side length $L$ in
$d$-dimensional Euclidean space $\R^{d}$, with $d\in\N$.
The single-particle Hamiltonians of the two systems 
differ by a local perturbation potential $V$. As a signature of
inequivalent representations of the canonical commutation relations,
the overlap $\boldsymbol\langle
\Phi^N_{L},\Psi^N_{L}\boldsymbol\rangle$ of the $N$-Fermion ground
states $\Phi^N_{L}$ and $\Psi^N_{L}$ must vanish in the thermodynamic
limit $L\to\infty$, $N\to\infty$, $N/L^{d} \to \text{const.} >0$
\cite[Chap.\ IV]{Fri53}, \cite[Chap.\ II.1.1]{Haa96}. A quantitative
version of this behaviour in terms of a power law
\begin{equation} \label{catastrophe}
  \bigabs{\boldsymbol\langle \Phi^N_{L},\Psi^N_{L}\boldsymbol\rangle}^2
  \sim L^{-\gamma}
\end{equation}
was predicted by P.\ W.\ Anderson in 1967. In
\cite{PhysRevLett.18.1049} he presented a brief computation for the
case of a point-like perturbation $V$ in $d=3$ dimensions and arrived at the
upper bound
\begin{equation} \label{exp-67a}
  \bigabs{\boldsymbol\langle \Phi^N_{L},\Psi^N_{L}\boldsymbol\rangle}^2
  \le L^{-\gamma_1}
\end{equation}
with
\begin{equation} \label{gamma-and-a}
  \gamma_1 = \pi^{-2} (\sin \delta)^{2}.
\end{equation}
Here, $\delta$ is the (single-particle) scattering phase shift caused
by the point interaction at the Fermi energy. Nowadays, this behaviour
is often referred to as Anderson's orthogonality catastrophe in the
physics literature.
A mathematical proof for a generalisation of \eqref{exp-67a} and
\eqref{gamma-and-a} was given recently in \cite{GKM}.
Allowing for a bounded, compactly
supported,  non-negative perturbation $V$ in $\R^{d}$, it is shown there that
\eqref{exp-67a} holds with
\begin{equation} \label{gamma-s-mat}
  \gamma_1 =  \frac{1}{\pi^{2}} \,\norm{T_{E}/2}_\HS^2,
\end{equation}
where $T_{E}$ denotes the transition matrix of scattering theory and
$\norm{\dotid}_{\HS}$ the Hilbert-Schmidt norm for operators on the
Hilbert space of the energy shell corresponding to the Fermi energy
$E$.
In the special case considered in \cite{PhysRevLett.18.1049},
\eqref{gamma-s-mat} reduces to \eqref{gamma-and-a}. The principal
strategy of the argument in \cite{GKM} is to rewrite the overlap
determinant as $\abs{\boldsymbol\langle
\Phi^N_{L},\Psi^N_{L}\boldsymbol\rangle}^{2} = \det A = \exp (\tr\ln
A)$ and to expand the logarithm in a series of non-negative terms
\begin{equation} \label{log-series}
  \bigabs{\boldsymbol\langle
  \Phi^N_{L},\Psi^N_{L}\boldsymbol\rangle}^{2}
  =
  \exp\bigg\{ - \sum_{n\in\N} \frac{1}{n} \, \tr \bigparens{(I-A)^{n}} \bigg\},
\end{equation}
see Lemma~\ref{lemma:overlap_series} below.
A similar idea was used by M.\ Kac \cite{Kac1954} in his proof of the
Szeg\longHungarianUmlaut{o}
limit theorem for Toeplitz determinants which is, in a way, an analogue
to \eqref{catastrophe}.

By dropping all but the first term $\tr (I -A)$ of the series,
which is called \textdef{Anderson integral} in the physics literature, one
arrives at an upper bound. The main work of \cite{GKM} consists in
deriving a lower bound of the form $\tr (I -A) \ge \gamma_1
\ln L$ for the Anderson integral with $\gamma_1$ given by
\eqref{gamma-s-mat}.
There are only few other mathematically rigorous works on Anderson's
orthogonality catastrophe \cite{KuOtSp13, Geb15a, KnOtSp15, Geb15b}. It is shown in \cite{KuOtSp13} that
\eqref{gamma-s-mat} in fact provides the exact coefficient in the
asymptotics $\tr (I -A) \sim \gamma_1 \ln L$ of the Anderson
integral in the thermodynamic limit for one-dimensional systems.
We refer to \cite{KuOtSp13, GKM} and references therein for a brief
description of the relevance of the orthogonality catastrophe in
physics and for a discussion of the theoretical approaches in the
physics literature.

In a second paper \cite{PhysRev.164.352} in 1967, P.\ W.\ Anderson
notes as an aside that the true asymptotics \eqref{catastrophe} of the
overlap involves an exponent $\gamma$ for which ``\ldots{} the main
difference from the previous result [i.e.\ \eqref{gamma-and-a}] is to
replace $(\sin \delta)^{2}$ by $\delta^{2}$.''
After some controversies about the correctness of interchanging limits
\cite{RiSi71,Ham71}, Anderson's result \eqref{catastrophe} was confirmed
in the case of a point interaction $V$ with the decay exponent
\begin{equation} \label{gamma-and-b}
  \gamma = \pi^{-2} \delta^{2}
\end{equation}
by theoretical-physics methods \cite{Ham71}. A mathematical proof was given recently in \cite{Geb15b}. For reasons of
comparison, we remark that the particle number $N$ in \cite{Ham71}
refers to the number of $s$-orbital states below the Fermi energy and
thus $N \sim L$. Related results in the context of the Kondo problem
in the physics literature can be found in \cite{PhysRev.178.1097,
  PhysRevB.1.1522}.

The purpose of the present paper is a mathematical contribution
towards the exact asymptotics \eqref{catastrophe}. We will prove in
Theorem~\ref{theorem:catastrophe} that, in the presence of a rather
general background potential $V_{0}$, a bounded, compactly
supported, non-negative perturbation potential $V$ in $\R^{d}$ causes the power-law
decay
\begin{equation} \label{exp-67b}
  \bigabs{\boldsymbol\langle \Phi^N_{L},\Psi^N_{L}\boldsymbol\rangle}^2
  \le L^{-\gamma + \oh(L^{0})}
\end{equation}
of the overlap for almost every Fermi energy $E \in \R$ along
subsequences $L \to\infty$. The decay exponent is given by
\begin{equation} \label{gamma-b-t}
  \gamma = \frac{1}{\pi^2}\bignorm{\arcsin\abs{T_E/2}}_{\HS}^2.
\end{equation}
We refer to Theorem~\ref{theorem:catastrophe} for the precise
statement. In proving~\eqref{gamma-b-t}, we obtain a result on the
trace of a product of spectral projections of two Schr\"odinger
operators which may be interesting by itself, see
Theorem~\ref{theorem:nth-term}.

Clearly, when comparing \eqref{gamma-b-t} to \eqref{gamma-s-mat}, we
infer $\gamma_1 \le \gamma$, and the two exponents are related
in the spirit of Anderson's rule quoted above. In view of \cite{Geb15b} and of the
physicists' results, we conjecture
that the exponent $\gamma$ governs the true asymptotics
\eqref{catastrophe} of the overlap whenever the modulus of the (appropriately defined) scattering phases does not exceed $\pi/2$.

The proof of Theorem~\ref{theorem:catastrophe} relies on the
representation \eqref{log-series} of the overlap.
We determine the dominant behaviour of each term in the $n$-sum in
\eqref{log-series}, because each term contributes to the asymptotics.
In order to treat the terms with $n > 1$ we have to deal with
additional issues. One is the non-positivity of certain trace
expressions, another one is to compute the multi-dimensional integral
\begin{equation} \label{hilbert-int}
  \int_{(0,\infty)^{2n}}\du_{1}\dotsi\du_{2n}\,
  \frac{e^{-(u_{1} + \dotsb + u_{2n})}}{(u_1 + u_{2})
    \dotsm (u_{2n-1}+ u_{2n})},
\end{equation}
which contributes to the asymptotics of the $n$th term in
\eqref{log-series}.
Subsequently, the values of these integrals show up in the Taylor expansion of
the function $x\mapsto (\arcsin x)^2$. We compute the integral
\eqref{hilbert-int} in Sect.~\ref{ssec:u-integral} by identifying it
with the first diagonal matrix element of the $(2n-1)$th
power of the Hilbert matrix.

Since $V$ causes scattering, the exponent $\gamma$ is typically
expected to be strictly positive. In the appendix, we prove this in the
case without a background potential.

After we completed this paper, Pushnitski and Frank \cite{FrPu15} established results on the asymptotics for traces of regularised projections of infinite-volume operators. Their work is partly a generalisation of our analysis in Sections~\ref{ssec:undoing-the-smoothing} to~\ref{ssec:u-integral}. In particular, their consequent use of Hankel operators is conceptually valuable and leads to a simplification of proofs. From this point of view it is also less surprising that (a unitary equivalent operator to) the Hilbert matrix appears in our Section~\ref{ssec:u-integral} when we compute the multi-dimensional integral \eqref{hilbert-int}.

%
\section{Setup and main result}
\label{sec:setup}

Let $d\in\N$, $\Lambda_1\subseteq\R^d$ be open and bounded with
$0\in\Lambda_1$ and for $L > 1$, define $\Lambda_L := L\cdot\Lambda_1$.

Let the negative Laplacian $-\Delta_L$ be supplied with Dirichlet
boundary conditions on $\Lambda_L$. We define two multiplication
operators $V_0$ and $V$ acting on $L^2(\Lambda_L)$, corresponding to
real-valued functions on $\R^d$ with the properties
\begin{equation*} \label{eq:assumption:V_V0}
  \tag{\textsf{V}}
  \begin{gathered}
    \max\set{V_0, 0}\in K_\loc(\R^d),
    \quad
    \max\set{-V_0, 0}\in K(\R^d), \\
    V\in L^\infty(\R^d),
    \quad
    V\ge 0,
    \quad
    \spt V\subseteq\Lambda_1\text{ compact.}
  \end{gathered}
\end{equation*}
Here, we have written  $ K(\R^d)$ and $K_\loc(\R^d)$
for the Kato class and the local Kato class, respectively, see
\cite{MR670130}.
The \textdef{finite-volume one-particle Schr\"odinger operators}
$H_L := -\Delta_L + V_0$ and $H_L' := H_L + V$ are self-adjoint and
densely defined in the Hilbert space $L^2(\Lambda_L)$. The
\textdef{infinite-volume} operators $H := -\Delta + V_0$ and $H' := H
+ V$ are self-adjoint and densely defined in the Hilbert space
  $L^2(\R^d)$. Birman's theorem, see
  \cite[Thm.~2]{BiEn67e} or \cite[Thm.~XI.10]{MR529429}, is applicable by virtue of 
  \cite[Thm.~B.9.1]{MR670130} and guarantees the
existence and completeness of the wave operators for the pair $H,
H'$. In particular, their absolutely continuous spectra are the same,
i.e.
\begin{equation}
  \spec_\ac(H) = \spec_\ac(H').
\end{equation}

The assumptions \eqref{eq:assumption:V_V0} on $V_0$ and $V$, together
with \cite[Thm.~6.1]{MR1756112}, imply that
the semigroup operators
$e^{-tH_L}$ and $e^{-tH_L'}$ generated by the finite-volume one-particle
operators $H_L$ and $H_L'$ are trace class for every $t > 0$, and, a
fortiori, compact. In particular, $H_L$ and
$H_L'$ are bounded from below and have purely discrete spectra.
We write $\lambda_1^L\le\lambda_2^L\le\dotsb$ and
$\mu_1^L\le\mu_2^L\le\dotsb$ for their non-decreasing sequences of
eigenvalues, counting multiplicities, and $(\varphi_j^L)_{j\in\N}$ and
$(\psi_k^L)_{k\in\N}$ for the corresponding sequences of normalised
eigenfunctions with an arbitrary choice of basis vectors in any
eigenspace of dimension greater than one.

Given $N\in\N$, the \textdef{induced} (\textdef{non-interacting})
\textdef{finite-volume $N$-particle Schr\"odinger operators} $\hat{H}_L$ and
$\hat{H}_L'$ act on the totally antisymmetric subspace
$\bigwedge_{j=1}^N L^2(\Lambda_L)$ of the $N$-fold tensor product space
and are given by
\begin{equation}
  \hat{H}_L^{(\prime)}
  :=
  \sum_{j=1}^N
  I^{\,\otimes^{j-1}} \otimes H_L^{(\prime)} \otimes I^{\,\otimes^{N-j}}.
\end{equation}
The corresponding ground states
are given by the totally antisymmetrised products
\begin{equation}
  \Phi^N_L
  :=
  \frac{1}{\sqrt{N!}}\,\varphi_1^L\wedge\dotsm\wedge\varphi_N^L,
  \qquad
  \Psi^N_L :=
  \frac{1}{\sqrt{N!}}\,\psi_1^L\wedge\dotsm\wedge\psi_N^L.
\end{equation}

In order to avoid ambiguities from possibly degenerate eigenspaces and
to realise a given \textdef{Fermi energy} $E\in\R$ in the thermodynamic limit,
we choose the \textdef{number of particles} as
\begin{equation} \label{eq:N}
  N_L(E) := \#\set{j\in\N\st \lambda_j^L \le E}\in\N_0,
\end{equation}
which is the eigenvalue counting function of $H_L$ at $E$.

The quantity of interest is the \textdef{ground-state overlap}
\begin{equation} \label{eq:overlap}
  \S_L(E)
  :=
  \Bigl\<\Phi_L^{N_L(E)}, \Psi_L^{N_L(E)}\Bigr\>_{N_L{(E)}}
  =
  \det\Bigparens{\<\varphi_j^L, \psi_k^L\>}_{j,k=1,\dotsc,N_L(E)},
\end{equation}
in particular its asymptotic behaviour as $L\to\infty$. In
\eqref{eq:overlap}, $\<\dotid,\dotid\>_N$ stands for the
scalar product on the $N$-fermion space $\bigwedge_{j=1}^N
L^2(\Lambda_L)$, and $\<\dotid,\dotid\>$ for the one on the
single-particle space $L^{2}(\Lambda_L)$.
If $N_L(E)=0$, we set $\S_L(E) := 1$.

\begin{remark} \label{remark:ids}
  The particular choice \eqref{eq:N} of $N_L(E)$ as an eigenvalue
  counting function turns out to be technically
  useful when conducting the thermodynamic limit, see
  Lemma~\ref{lemma:trace_with_spec_proj} below. The \textdef{particle density} $\rho(E)$ of the
  two non-interacting fermion systems in the thermodynamic limit
  coincides with the
  integrated density of states
  \begin{equation} \label{eq:density}
    \rho(E) = \lim_{L\to\infty} \frac{N_L(E)}{L^d\, \abs{\Lambda_1}}
  \end{equation}
  of the single-particle Schr\"odinger operator $H$ (which is the same
  as the integrated density of states of $H'$),
  provided the limit exists.
  Here, $\abs{\Lambda_1}$ denotes the Lebesgue measure of
  $\Lambda_1\subseteq\R^d$. Situations where the limit
  \eqref{eq:density} is known to exist include
  periodic $V_0$, or $V_0$ vanishing at infinity. If the limit
  \eqref{eq:density} does not exist, then this is due to the occurrence of more than one
  accumulation point, because the assumptions on $V_0$ in
  \eqref{eq:assumption:V_V0}, together
  with~\cite[Thm.~C.7.3]{MR670130}, imply
  $\limsup_{L\to\infty}N_L(E)/L^d < \infty$ for every $E\in\R$.
  We will study the asymptotic
  behaviour of the overlap $\S_L(E)$ as $L\to\infty$ regardless of the
  existence of the limit \eqref{eq:density}.
\end{remark}

The main result of this paper is an upper bound on the ground-state
overlap $\S_L(E)$ for large $L$.
Throughout we use the convention $\ln 0 := -\infty$. The terms null set
and almost-every (a.e.) refer to Lebesgue measure if not
specified otherwise.

\begin{theorem}[Orthogonality Catastrophe] \label{theorem:catastrophe}
  Assume conditions \eqref{eq:assumption:V_V0}.
  Let $(L_m)_{m\in\N}$ be a sequence in
  $(0,\infty)$ with $L_m\to\infty$. Then there exist a subsequence
  $(L_{m_k})_{k\in\N}$, a null set $\Nn\subseteq\R$ of exceptional Fermi energies and a function
  $\gamma\from\R\setminus\Nn\to[0,\infty)$ such that for every
  $E\in\R\setminus\Nn$ the ground-state overlap~\eqref{eq:overlap} obeys
  \begin{equation} \label{eq:catastrophe}
    \abs{\S_{L_{m_k}}(E)}
    \le
    \exp\Bigparens{-\tfrac{1}{2}\gamma(E)\ln L_{m_k} + \oh(\ln L_{m_k})}
    = L_{m_k}^{-\gamma(E)/2 + \oh(1)}
  \end{equation}
  as $k\to\infty$.
  Equivalently,
  \begin{equation} \label{eq:catastrophe-without-a}
    \limsup_{k\to\infty}\frac{\ln\abs{\S_{L_{m_k}}(E)}}{\ln L_{m_k}}
    \le
    -\frac{\gamma(E)}{2}.
  \end{equation}
  The decay exponent $\gamma$ is given by
  \begin{equation} \label{eq:gamma}
    \gamma(E)
    :=
    \frac{1}{\pi^2}\bignorm{\arcsin\abs{T_E/2}}_{\HS}^2.
  \end{equation}
  Here, $T_E := S_E - I_E$ is the transition matrix, $S_E$ is the
  scattering matrix for the pair $(H,H')$ and energy $E$, and
  $\norm{\dotid}_\HS$ denotes the Hilbert--Schmidt norm on the fibre
  Hilbert space $\H_E$, on which $T_E$ and $S_E$ are defined.
\end{theorem}

\begin{remarks}
  \item We refer to Subsection~\ref{ssec:scattering-theory} for a more
        precise definition of the scattering-theoretic quantities
        $T_E$ and $S_E$.
  \item In proving Theorem~\ref{theorem:catastrophe}, we obtain a
        result on the asymptotics of the trace
        $\tr\bigbraces{\bigparens{1_{(-\infty,E]}(H_L)\1_{(E,\infty)}(H_L')}^n}$
        as $L\to\infty$, which may be interesting by itself; see
        Theorem~\ref{theorem:nth-term}.
  \item \label{no-subseq}
  The reason for passing to a subsequence $(L_{m_k})_{k\in\N}$ in Theorem~\ref{theorem:catastrophe}
  originates from Lemma~\ref{lemma:trace_with_spec_proj} below. What stands behind it is the lack of
  known a.e.-bounds on the
  finite-volume spectral shift function for the pair of operators $H_{L}, H_{L}'$, which hold
  uniformly in the limit $L\to\infty$. This unfortunate fact has been noticed many times
  in the literature, see e.g.\ \cite{MR2596053}, and the pathological behaviour of the spectral shift function found in
  \cite{MR908658} illustrates that this is a delicate issue.
  However, in certain special situations such a.e.-bounds are known, and our result can be
  strengthened. More precisely, we have
\end{remarks}

\noindent
  \textbf{Theorem~\ref{theorem:catastrophe}'.}~ {\itshape
  Assume the situation of Theorem~\ref{theorem:catastrophe} with
  $d=1$, or replace the perturbation potential $V$ in
  Theorem~\ref{theorem:catastrophe} by a finite-rank operator
  $V = \sum_{\nu=1}^{n} \<\phi_{\nu}, \dotid\> \, \phi_{\nu}$ with
  compactly supported $\phi_{\nu} \in L^{2}(\R^{d})$ for
  $\nu=1,\ldots,n$, or consider the lattice problem on $\Z^d$
  corresponding to the situation in
  Theorem~\ref{theorem:catastrophe}.
  Then the ground-state overlap \eqref{eq:overlap} obeys
  \begin{equation}
    \abs{\S_L(E)}
    \le
    \exp\Bigparens{-\tfrac{1}{2}\gamma(E)\ln L + \oh(\ln L)}
    = L^{-\gamma(E)/2 + \oh(1)}
  \end{equation}
  for a.e.\ $E\in\R$ as $L\to\infty$. Equivalently,
  \begin{equation}
    \limsup_{L\to\infty} \frac{\ln\abs{\SO_{L}(E)}}{\ln L} \le - \frac{\gamma(E)}{2}
  \end{equation}
  for a.e.\ $E\in\R$.}
\medskip

\begin{remarks}
\item In \cite{GKM}, similar statements to
      Theorem~\ref{theorem:catastrophe} and
      Theorem~\ref{theorem:catastrophe}' were proved,
      in particular, the bound
      \begin{equation}\label{eq:old result}
        \limsup_{L\to\infty} \frac{\ln\abs{\SO_{L}(E)}}{\ln L} \le - \frac{\gamma_1(E)}{2},
      \end{equation}
      with the exponent
      \begin{equation}\label{eq:gamma2}
        \gamma_1(E)= \frac{1}{\pi^2}\bignorm{T_E/2}_{\HS}^2.
      \end{equation}
      Note that $\gamma_1(E)$, which is called $\gamma(E)$ in
      \cite{GKM}, is strictly smaller than
      $\gamma(E)$ whenever both are non-zero. The bigger exponent $\gamma(E)$
      is due to treating all terms in a series
      expansion of $\ln\abs{\S_L(E)}$  (see equation~\eqref{eq:overlap_series}
      below) instead of only the \textdef{Anderson integral}, which is the first
      term of the series and gives rise to $\gamma_1(E)$.
    \item \label{remark:otte:first} Another mathematical work dealing
          with AOC is \cite{KuOtSp13}. That paper proves the \emph{exact} asymptotics of the
      Anderson integral in the special case $d=1$ and
      $V_0=0$. In particular, this
      yields a bound on the overlap as in \eqref{eq:old result} with
      the same non-optimal $\gamma_1(E)$ given by
      \eqref{eq:gamma2}.
      The paper also provides a lower bound on
      $\SO_L(E)$ with a smaller decay exponent \cite[Cor.~5.6]{KuOtSp13}.
\end{remarks}

\section{Series expansion of the overlap}
\label{sec:series-expansion}

In order to expand the ground-state overlap as a series, we introduce
the orthogonal projections
\begin{equation} \label{eq:projections_till_N}
  P^N_L := \sum_{j=1}^N \<\varphi_j^L, \dotid\>\varphi_j^L
  \quad\text{and}\quad
  \Pi^N_L := \sum_{k=1}^N \<\psi_k^L, \dotid\>\psi_k^L
\end{equation}
for $N\in\N_0$,
i.e.\ the projections on the eigenspaces of the first $N$ eigenvalues.
Using those, we can prove the following lemma.

\begin{lemma} \label{lemma:overlap_series}
  Let $L > 1$, $E\in\R$ and assume that $\S_L(E) \ne 0$. Then
  \begin{equation} \label{eq:overlap_series}
    \abs{\S_L(E)}^2
    =
    \exp\biggparens{
      -\sum_{n=1}^\infty
      \frac{1}{n}
      \tr\Bigbraces{
        \Bigparens{P^{N_L(E)}_L\bigparens{I-\Pi^{N_L(E)}_L}}^n
    }},
  \end{equation}
  where we take the trace of operators on the Hilbert space $L^2(\Lambda_L)$.
  \begin{proof}
    For brevity, set $N := N_L(E)$. If $N = 0$, the assertion is true
    by definition. Otherwise,
    define the $N\times N$-matrix $M := \bigparens{\<\varphi_j^L,
    \psi_k^L\>}_{j,k=1,\dotsc,N}$. Then $\S_L(E) = \det M$ and
    $\abs{\S_L(E)}^2 = \det({MM}^*)$.
    For $1\le j,\ell\le N$, the
    $(j,\ell)$-th entry of ${MM}^*$ is
    \begin{equation}
      (MM^*)_{j,\ell}
      =
      \sum_{k=1}^N
      \<\varphi_j^L, \psi_k^L\>\<\psi_k^L,\varphi_\ell^L\>
      =
      \<\varphi_j^L, \Pi_L^N\varphi_\ell^L\>= \<\varphi_j^L, P_L^N\Pi_L^NP_L^N \varphi_\ell^L\>.
    \end{equation}
    Since $\S_L(E)\ne 0$ by assumption and therefore 
    ${MM}^* > 0$, we have
    $0\leq P_L^N(I-\Pi_L^N)P_L^N<1$. Moreover, being of finite rank, $P_L^N(I-\Pi_L^N)P_L^N$
    is a trace class operator.
    Thus, we compute
    \begin{align}
      \abs{\S_L(E)}^2
      & =
      \det \left( I - P_L^N(I-\Pi_L^N)P_L^N\right) \notag\\
      & =
      \exp\Bigparens{
        \tr\Bigbraces{
          \ln\bigparens{I-P_L^N(I-\Pi_L^N)P_L^N}}} \notag\\
      & =
      \exp\biggparens{-\tr\biggbraces{\sum_{n=1}^\infty \frac{1}{n}
        \Bigparens{P_L^N(I-\Pi_L^N)P_L^N}^n}} \notag\\
      & =
      \exp\biggparens{-\sum_{n=1}^\infty \frac{1}{n}
        \tr\Bigbraces{\Bigparens{P_L^N(I-\Pi_L^N)}^n}},
    \end{align}
    where we used the expansion $\ln(1-x) = -\sum_{n=1}^\infty x^n/n$
    for the logarithm, which converges absolutely for $\abs{x} < 1$.
  \end{proof}
\end{lemma}

\begin{remark}
  Lemma~\ref{lemma:overlap_series} will be the starting point of our
  estimates for $\abs{\S_L(E)}$. Equation~\eqref{eq:overlap_series}
  can be written as
  \begin{equation} \label{eq:log_overlap_series}
    - \ln\abs{\S_L(E)}
    =
    \frac{1}{2}\sum_{n=1}^\infty
    \frac{1}{n}
    \tr\Bigbraces{
      \Bigparens{P^{N_L(E)}_L(I-\Pi^{N_L(E)}_L)}^n
    }.
  \end{equation}
  The trace expressions in \eqref{eq:log_overlap_series} are
  non-negative, so any truncation of the series yields a lower bound
  on $-\ln\abs{\S_L(E)}$, and therefore an upper bound on the
  overlap. Keeping only the term for $n=1$, one recovers the so-called
  \textdef{Anderson integral}, which was estimated in \cite{GKM}.

  In the sequel, we will find an upper bound on $\abs{\S_L(E)}$ by
  bounding each individual term of \eqref{eq:log_overlap_series} from
  below.
\end{remark}

We begin by recasting the orthogonal projections
\eqref{eq:projections_till_N} as functions of $H_L$ and $H_L'$ in the
sense of the spectral calculus.
The projections in
\eqref{eq:projections_till_N} are not necessarily functions of $H_L$
and $H_L'$, since the $N$th eigenvalues might
be of multiplicity higher than one.
The choice of $N_L(E)$ in \eqref{eq:N}, together with a
convergence result of the spectral shift function, allows us to
rewrite them, at the cost of passing to a
subsequence of lengths.

\begin{lemma} \label{lemma:trace_with_spec_proj}
  For $n\in\N$, $L > 1$ and $E\in\R$, define
 \begin{equation} \label{eq:trace_function_of_HL}
   \F^n_{L}(E) :=
   \tr\Bigbraces{\Bigparens{\1_{(-\infty,E]}(H_L)\1_{(E,\infty)}(H_L')}^n
  }
 \end{equation} and
 \begin{equation} \label{eq:P-Pi-power}
   \I^n_{L}(E) :=
   \tr\Bigbraces{\Bigparens{P^{N_{L}(E)}_{L}(I-\Pi^{N_{L}(E)}_{L})}^n
   }.
 \end{equation}
 Then
 \begin{nummer}
 \item\label{ssf}
   Assume \eqref{eq:assumption:V_V0}
   and let $(L_m)_{m\in\N} \subset (0,\infty)$ be a
   sequence of increasing lengths with $L_m\upto\infty$.
   Then
   there exists a subsequence $(L_{m_k})_{k\in\N}$ 
   such that for a.e.\ Fermi energy $E\in\R$
   \begin{equation} \label{eq:trace_E_fixed}
     \bigabs{\F^n_{L_{m_k}}(E)-\I^n_{L_{m_k}}(E)}
     =
     \oh(\ln L_{m_k})
   \end{equation}
   as $k\to\infty$. 
 \item Assume the situation of Theorem~{\upshape\ref{theorem:catastrophe}'}. Then
   \begin{equation}
     \sup_{L>1}\sup_{E\in\R} \bigabs{\F^n_{L}(E) - \I^n_{L}(E)} 	< \infty.
   \end{equation}
 \end{nummer}
  \begin{proof}
    For fixed $L > 1$ and $E\in\R$, the definition of $N_L(E)$ in
    $\eqref{eq:N}$ implies
    \begin{equation}
      \lambda_{N_L(E)}^L \le E < \lambda_{N_L(E)+1}^L
      \le \mu_{N_L(E)+1}^L
    \end{equation}
    if we set $\lambda_0^L := -\infty$. This allows us to write
    \begin{equation} \label{eq:P-equals-indicator}
      P_L^{N_L(E)} = \1_{(-\infty,E]}(H_L)
    \end{equation}
    and
    \begin{align} \label{eq:ssf-perturb-sum}
      I-\Pi^{N_L(E)}_L
      & =
      \1_{(E,\infty)}(H_L')
      -
      \sum_{k=1}^{N_L(E)} \1_{(E,\infty)}(\mu_k^L)\,
      \<\psi_k^L,\dotid\>\psi_k^L \notag\\
      & =:
      \1_{(E,\infty)}(H_L')
      -
      Q.
    \end{align}
    The operator $Q$ is an orthogonal projection with trace
    \begin{align}
      \tr Q & = \#\bigl\{k\in\set{1,\dotsc,N_L(E)}\st \mu_k^L > E\bigr\}
      \notag\\ & =
      N_L(E) - \#\bigl\{k\in\N\st \mu_k^L \le E\bigr\}
      =: \xi_L(E)
    \end{align}
    equal to the finite-volume spectral-shift function at the Fermi energy. 
    
    Using $A^n - B^n = \sum_{k=1}^n B^{k-1}(A-B)A^{n-k}$ for bounded
    operators $A$ and $B$, we write the difference of operator powers
    on the left-hand side of \eqref{eq:trace_E_fixed} as
    \begin{multline} \label{eq:difference-of-powers}
      \Bigparens{P^{N_{L}(E)}_{L}\1_{(E,\infty)}(H_{L}')}^n
        -
      \Bigparens{P^{N_{L}(E)}_{L}(I-\Pi^{N_{L}(E)}_{L})}^n
      \\
      = \sum_{k=1}^n
      \Bigparens{P^{N_{L}(E)}_{L}(I-\Pi^{N_{L}(E)}_{L})}^{k-1}
      P^{N_{L}(E)}_{L}
      \, Q \,
      \Bigparens{P^{N_{L}(E)}_{L}\1_{(E,\infty)}(H_{L}')}^{n-k},
    \end{multline}
    where we also use~\eqref{eq:P-equals-indicator}. We
    estimate the traces of the operators on the right-hand side
    of~\eqref{eq:difference-of-powers} by bounding the operator norms
    of all projections, except for $Q$, by $1$. We then arrive at $n
    \xi_L(E)$ as a upper bound for \eqref{eq:difference-of-powers}.
    The claim follows by exploiting the weak convergence of $\xi_L$ as $L\to\infty$
    \cite[Thm.~1.4]{MR2596053} in the situation of (i), or using the
    uniform boundedness of $\xi_L$ in the situation of (ii). We refer to
    \cite[Lemma~3.9]{GKM} for a detailed argument.
  \end{proof}
\end{lemma}

Having established \eqref{eq:trace_E_fixed}, we will prove a diverging
lower bound for $\tr\bigbraces{\bigparens{\1_{(-\infty,E]}(H_L)
\1_{(E,\infty)}(H_L')}^n}$ as $L\to\infty$. There will be no
restriction to particular sequences of lengths from now on.
The following theorem is the main ingredient of the proof.

\begin{theorem}\label{theorem:nth-term}
  Assume the situation of Theorem~\ref{theorem:catastrophe} or
  Theorem~\ref{theorem:catastrophe}'.
  Then there exists a null set
  $\mathcal{N} \subset\R$ of exceptional Fermi energies such that
  \begin{equation} \label{eq:nth-term}
    \tr\bigbraces{\bigparens{1_{(-\infty,E]}(H_L)\1_{(E,\infty)}(H_L')}^n}
      \ge
      n J_{2n} \tr(\abs{T_E/(2\pi)}^{2n}) \ln L + \oh(\ln L)
  \end{equation}
  for every $E\in\R\setminus\Nn$ and every $n\in\N$
  as $L\to\infty$. The error term $\oh(\ln L)$ depends on $n$ and
  $E$, and we introduced the constant
  \begin{equation} \label{eq:J2n-def}
    J_{2n} := \pi^{2(n-1)}2^{2n-1}\frac{[(n-1)!]^2}{(2n)!}.
  \end{equation}
\end{theorem}

\begin{remarks} \label{remarks:nth-term}
  \item In the next section, we will spell out explicitly the proof of
        Theorem~\ref{theorem:nth-term} for the situation of
        Theorem~\ref{theorem:catastrophe} only. It follows from
        Corollary~\ref{corollary:the-asymptotics},
        Theorem~\ref{theorem:u-integral} and Theorem~\ref{thm:scattering}.
        The proof is fully
        analogous (and even simpler) in the remaining  situations of
        Theorem~\ref{theorem:catastrophe}', where $V$ is a finite-rank
        operator.
  \item The constant $J_{2n}$ will emerge as the value of a
        $2n$-dimensional integral which we calculate using the
        spectral representation of the Hilbert matrix, see
        Subsection~\ref{ssec:u-integral} below.
\end{remarks}

Given Theorem~\ref{theorem:nth-term}, we are now in a position to prove
Theorem~\ref{theorem:catastrophe}.

\begin{proof}[Proof of Theorem~\ref{theorem:catastrophe}]
  Let $M\in\N$. Let $\Nn$ be the null set from
  Theorem~\ref{theorem:nth-term}. Let $E\in\R\setminus\Nn$.
  We start from Lemma~\ref{lemma:overlap_series} and
  Lemma~\ref{lemma:trace_with_spec_proj}, which imply
  \begin{align}
    -\ln\abs{\S_{L_{m_k}}(E)}
    & \ge
    \frac{1}{2}
    \sum_{n=1}^M \frac{1}{n}
    \tr\Bigbraces{
      \Big(P^{N_{L_{m_k}}(E)}_{L_{m_k}} \big(I-\Pi^{N_{L_{m_k}}(E)}_{L_{m_k}}\big)\Big)^n
    } \notag \\
    & =
    \frac{1}{2}
    \sum_{n=1}^M \frac{1}{n}
    \tr\Bigbraces{
      \Bigparens{\1_{(-\infty,E]}(H_{L_{m_k}})\1_{(E,\infty)}(H_{L_{m_k}}')}^n
    }   \notag \\
      & \quad + \oh(\ln L_{m_k}) 
  \end{align}
  for a subsequence $(L_{m_k})_{k\in\N}$, as $k\to\infty$, with an
  $M$-dependent error term $\oh(\ln L_{m_k})$.
  By Theorem~\ref{theorem:nth-term}, this gives  
	\begin{equation}
     -\ln\abs{\S_{L_{m_k}}(E)}
    \ge
      \frac{1}{2} \tr\sum_{n=1}^M J_{2n}\hairspace
      \abs{T_E/(2\pi)}^{2n}
      \ln L_{m_k} + \oh(\ln L_{m_k}) \label{eq:catastrophic-calculation-0}
	\end{equation}
  as $k\to\infty$, with an $M$-dependent error term
  $\oh(\ln L_{m_k})$.
  The constants $J_{2n}$ show up in the series expansion
  \cite[Eq.~1.645\,2]{MR2360010}
  \begin{equation}
    \sum_{n=1}^\infty  J_{2n} x^{2n}
    =
    \pi^{-2}\bigparens{\arcsin(\pi x)}^2
    \qquad\text{for } \abs{x} \le \tfrac{1}{\pi}
    .
   \end{equation}
  Therefore, monotone convergence and the functional calculus yield
  \begin{equation}
    \lim_{M\to\infty}\tr\sum_{n=1}^M J_{2n}\abs{T_E/(2\pi)}^{2n} =
    \pi^{-2}\bignorm{\arcsin\abs{T_E/2}}_\HS^2.
  \end{equation}
  Since \eqref{eq:catastrophic-calculation-0} is valid for every
  $M\in\N$, we infer
  \begin{equation}
    \limsup_{k\to\infty}\frac{\ln\abs{\S_{L_{m_k}}(E)}}{\ln L_{m_k}}
    \le
    -\frac 1 2 \pi^{-2}\bignorm{\arcsin\abs{T_E/2}}_\HS^2
    =
    -\frac{\gamma(E)}{2},
  \end{equation}
  which proves~\eqref{eq:catastrophe-without-a}. For
  \eqref{eq:catastrophe}, note that by the definition of the limit
  superior for every $\varepsilon > 0$ there is $k_0\in\N$ such that
  \begin{equation}
   \frac{\ln\abs{\S_{L_{m_k}}(E)}}{\ln L_{m_k}}
   \le
   -\frac{\gamma(E)}{2} + \varepsilon
  \end{equation}
  for all $k\geq k_0$, which implies the claim.
\end{proof}

It remains to prove Theorem~\ref{theorem:nth-term}.

\section{Proof of Theorem~\ref{theorem:nth-term}}
\label{sec:proof-of-theorem:nth-term}

\subsection{An integral representation for
\texorpdfstring{\boldmath$\tr\bigbraces{\bigparens{f(H_L)g(H_L')}^n}$}{traces}}

Throughout this subsection, $n\in\N$, $L > 1$ and $E\in\R$ are all fixed.
Using the eigenvalue equations of $H_L$ and $H_L'$, we rewrite
trace expressions like~\eqref{eq:trace_function_of_HL}.

\begin{lemma}
  \label{lemma:trace_power_projections_as_sum}
  Let $f,g\from\R\to [0,1]$ be measureable functions with compact
  supports $\spt f\subseteq (-\infty,E]$ and $\spt g\subseteq
  (E,\infty)$. Then
  \begin{multline} \label{eq:trace_power_as_sum}
    \tr\bigbraces{\bigparens{f(H_L)g(H_L')}^n}
        \\ =
      \sum_{\alpha,\beta\in\N^n} \prod_{j=1}^n
        \biggparens{
        f(\lambda_{\alpha_j}^L)
        g(\mu_{\beta_j}^L)
        \frac{\<\varphi_{\alpha_j}^L, V\psi_{\beta_j}^L\>
               \<\psi_{\beta_j}^L,V\varphi_{\alpha_{j+1}}^L\>}%
             {(\mu_{\beta_j}^L - \lambda_{\alpha_j}^L)(\mu_{\beta_j}^L-\lambda_{\alpha_{j+1}}^L)}},
    \end{multline}
  for multi-indices $\alpha=(\alpha_1,\dotsc,\alpha_n)\in\N^n$ with
  the convention $\alpha_{n+1} := \alpha_1$.
  \begin{proof}
    We begin noting that
    \begin{equation} \label{eq:projections_as_sums}
      f(H_L)
        = \sum_{j\in\N} f(\lambda_j^L)\,
          \<\varphi_j^L, \dotid\>\varphi_j^L,
      \qquad
      g(H_L')
        = \sum_{k\in\N} g(\mu_k^L)\,
          \<\psi_k^L, \dotid\>\psi_k^L.
    \end{equation}
    To ease notation, we employ the bra-ket notation in the next
    formula, writing $\<\varphi, \dotid\>\varphi =: \bra{\varphi}\ket{\varphi}$ for $\varphi\in
    L^2(\Lambda_L)$.
    Then \eqref{eq:projections_as_sums} implies
    \begin{equation}
      \bigparens{f(H_L)g(H_L')}^n
      =
      \sum_{\alpha,\beta\in\N^n}
        \Bigparens{\prod_{j=1}^n
          f(\lambda_{\alpha_j}^L)
          g(\mu_{\beta_j}^L)}
        \,
        \prod_{j=1}^n \bra{\varphi_{\alpha_j}^L}
                     \<\varphi_{\alpha_j}^L, \psi_{\beta_j}^L\>
                     \ket{\psi_{\beta_j}^L},
    \end{equation}
    and
    \begin{multline} \label{eq:trace_power_projections_with_bra_ket}
      \tr\bigbraces{\bigparens{f(H_L)g(H_L')}^n}
        \\ =
      \sum_{\alpha,\beta\in\N^n}
        \Bigparens{\prod_{j=1}^n
          f(\lambda_{\alpha_j}^L)
          g(\mu_{\beta_j}^L)}
        \,\prod_{j=1}^n
        \<\varphi_{\alpha_j}^L,\psi_{\beta_j}^L\>
        \<\psi_{\beta_j}^L,\varphi_{\alpha_{j+1}}^L\>,
    \end{multline}
    where we used the convention $\alpha_{n+1} := \alpha_1$ for
    $\alpha\in\N^n$.
    Now, we note that the eigenvalue equations imply
    \begin{equation}
      \lambda_j^L\<\varphi_j^L, \psi_k^L\>
      =
      \<H_L \varphi_j^L, \psi_k^L\>
      =
      \mu_k^L\<\varphi_j^L, \psi_k^L\> - \<\varphi_j^L, V\psi_k^L\>
    \end{equation}
    for $j,k\in\N$, and therefore
    \begin{equation} \label{eq:phi-psi-product}
      \<\varphi_j^L, \psi_k^L\>
      =
      \frac{\<\varphi_j^L, V\psi_k^L\>}{\mu_k^L - \lambda_j^L}
    \end{equation}
    whenever $\lambda_j^L\ne \mu_k^L$.
    Since $f$ and $g$ have disjoint supports,
    \eqref{eq:phi-psi-product} and
    \eqref{eq:trace_power_projections_with_bra_ket} yield the claim.
  \end{proof}
\end{lemma}

\begin{remark} \label{remark:finite_measure}
  In analogy to \cite{GKM}, one might be tempted to define a spectral
  correlation ``measure'' by
  \begin{multline} \label{eq:finite_measure}
    \mu_L^{2n}(A_1\times\dotsm\times A_n\times B_1\times\dotsm\times B_n)
    \\ :=
    \tr\big\{(\1_{A_1}(H_L)V\1_{B_1}(H_L')V\;\dotsm\; \1_{A_n}(H_L)V\1_{B_n}(H_L')V)\big\}
  \end{multline}
  for $n\in\N$, \ $L > 1$ and bounded $A_1,\dotsc, A_n,
  B_1,\dotsc,B_n\in\Borel(\R)$, which was done for the case $n=1$ in \cite{GKM}.
  Lemma~\ref{lemma:trace_power_projections_as_sum}
  would then read
  \begin{equation} \label{eq:trace_power_projections_as_integral_notused}
    \tr\bigbraces{\bigparens{f(H_L)g(H_L')}^n}
     =
    \int_{\R^n\times\R^n}\d\mu_L^{2n}(x,y)\prod_{j=1}^n
      \frac{f(x_j)g(y_j)}{(y_j - x_j)(y_j - x_{j+1})}.
  \end{equation}
  However, \eqref{eq:finite_measure} is
  not necessarily non-negative for $n\ge 2$, and therefore we cannot mimick the
  proof of \cite{GKM}.
\end{remark}

Next, we rewrite the right-hand side of
\eqref{eq:trace_power_as_sum} using a variation of an integral formula
that goes back to Feynman and Schwinger.

\begin{lemma}[Feynman--Schwinger parametrization]
  Let $x_1,\dotsc,x_n\in(0,\infty)$. Then
  \begin{equation} \label{eq:feynman_parametrization}
    \frac{1}{x_1\dotsm x_n}
    =
    \int_0^\infty\dt\, t^{n-1}
    \int_{(0,\infty)^n}\du\, \abs{u}_1e^{-\abs{u}_1} e^{-t u\cdot x},
  \end{equation}
  where $u\cdot x = \sum_{j=1}^n u_jx_j$ denotes the Euclidean scalar
  product and $\abs{u}_1 := \sum_{j=1}^n\abs{u_j}$ the $1$-norm on $\R^n$.
  \begin{proof}
    For any measurable function $f\from(0,\infty)^n\to (0,\infty)$ the
    coarea formula implies
    \begin{equation} \label{eq:coarea-formula}
      \int_{(0,\infty)^n}\du\, f(u)
      =
      \int_0^\infty\dt \int_{\M} \frac{\dS(\xi)}{\sqrt{n}} t^{n-1}f(t\xi),
    \end{equation}
    where $\dS$ stands for integration with respect to the surface
    measure on $\M := \set{\xi\in(0,\infty)^n\st \abs{\xi}_1 = 1}$.
    Let $r > 0$. Starting from $x_j^{-1} = \int_0^\infty\du_j\,
    e^{-u_jx_j}$, we compute using \eqref{eq:coarea-formula}
    \begin{align}
      \frac{1}{x_1\dotsm x_n}
      & =
      \int_{(0,\infty)^n}\du\, e^{-u\cdot x}
      =
      \int_0^\infty\dt \int_\M
      \frac{\dS(\xi)}{\sqrt{n}} t^{n-1}
      e^{-t \xi\cdot x} \notag\\
      & =
      \int_0^\infty\dt \int_\M
      \frac{\dS(\xi)}{\sqrt{n}} t^{n-1} r^n
      e^{-r t \xi\cdot x},
    \end{align}
    which is $r$-independent. Given any measurable function $g\from
    (0,\infty)\to(0,\infty)$ with $\int_0^\infty\dr\, \frac{g(r)}{r} = 1$,
    we therefore get
    \begin{align}
      \frac{1}{x_1\dotsm x_n}
      & =
      \int_0^\infty\dr\, g(r) \int_0^\infty\dt \int_\M
      \frac{\dS(\xi)}{\sqrt{n}} t^{n-1} r^{n-1}
      e^{-r t \xi\cdot x} \notag\\
      & =
      \int_0^\infty\dt\, t^{n-1}  \int_{(0,\infty)^n}\du\, g(\abs{u}_1)
      e^{-t u\cdot x},
    \end{align}
    where we used the Fubini--Tonelli theorem and
    \eqref{eq:coarea-formula} with $f(u) = g(\abs{u}_1)e^{-tu\cdot
    x}$. Choosing $g(r) := r e^{-r}$ finishes the proof.
  \end{proof}
\end{lemma}

We use \eqref{eq:feynman_parametrization} to rewrite the right-hand side of
\eqref{eq:trace_power_as_sum}.
\begin{lemma} \label{lemma:trace_using_feynman}
  Let $f,g\from\R\to [0,1]$ be measurable functions with compact
  supports $\spt f\subseteq (-\infty,E]$ and $\spt g\subseteq
  (E,\infty)$.
  Then,
  \begin{multline} \label{eq:Feynman_integrated_trace}
    \tr\bigbraces{\bigparens{f(H_L)g(H_L')}^n}
    \notag\\ 
    \quad =
    \int_0^\infty\dt\, t^{2n-1} \int_{(0,\infty)^n\times
      (0,\infty)^n}\d(u,v)\, (\abs{u}_1+\abs{v}_1)
      e^{-\abs{u}_1-\abs{v}_1} \hfill \notag\\
      \qquad \times
      \tr\bigg\{\prod_{j=1}^n
        \sqrt{V}f(H_L) e^{(u_j + v_{j-1})t(H_L-E)}V g(H_L') e^{-(u_j + v_j)t(H_L'-E)}\sqrt{V}\bigg\} \hfill
  \end{multline}
  with the convention $v_0 := v_n$ for $v\in\R^n$.
\end{lemma}

\begin{proof}
    Let $x\in (-\infty,0]^n$, $y\in(0,\infty)^n$ and define
    $x_{n+1} := x_1$. Then, by \eqref{eq:feynman_parametrization},
    \begin{multline}
      \frac{1}{\prod_{j=1}^n (y_j - x_j)(y_j - x_{j+1})}
      \\ \qquad =
    \int_0^\infty\dt\, t^{2n-1}
    \int_{(0,\infty)^n\times
      (0,\infty)^n}\d(u,v)\, (\abs{u}_1+\abs{v}_1)
      e^{-\abs{u}_1-\abs{v}_1}\\ \times
    \exp\biggparens{ -t\sum_{j=1}^n  \big( u_j(y_j-x_j) + v_j(y_j-x_{j+1})\big)}
    \end{multline}
    and
    \begin{equation}
      \sum_{j=1}^n \big( u_j(y_j-x_j) + v_j(y_j-x_{j+1}) \big)
      =
      \sum_{j=1}^n\bigparens{(u_j+v_j)y_j - (u_j+v_{j-1})x_j}
    \end{equation}
    for $u,v\in (0,\infty)^n$. Now, let $\alpha,\beta\in\N^n$. Setting $x_j =
    \lambda_{\alpha_j}^L - E$ and $y_j = \mu_{\beta_j}^L - E$, we can write the
    denominator in \eqref{eq:trace_power_as_sum} as
    \begin{multline} \label{eq:ev_denominator_using_feynman}
      \frac{1}{\prod_{j=1}^n
        (\mu_{\beta_j}^L - \lambda_{\alpha_j}^L)(\mu_{\beta_j}^L-\lambda_{\alpha_{j+1}}^L)}
      \\ \qquad =
      \int_0^\infty\dt\, t^{2n-1}
      \int_{(0,\infty)^n\times
      (0,\infty)^n}\d(u,v)\, (\abs{u}_1+\abs{v}_1)
      e^{-\abs{u}_1-\abs{v}_1}
      \\ \times
      \prod_{j=1}^n e^{-(u_j+v_j)t(\mu_{\beta_j}^L - E)}
                   e^{(u_j+v_{j-1})t(\lambda_{\alpha_j}^L - E)}.
    \end{multline}
    The sums over $\alpha$ and $\beta$ in
    \eqref{eq:trace_power_as_sum} contain only finitely many terms,
    due to the compact supports of $f$ and~$g$. Therefore these sums
    can be interchanged with the integrals
    from~\eqref{eq:ev_denominator_using_feynman}. This results in
    \begin{multline}
      \tr\bigbraces{\bigparens{f(H_L)g(H_L')}^n} \\
      =
      \int_0^\infty\dt\, t^{2n-1}
      \int_{(0,\infty)^n\times (0,\infty)^n}\d(u,v)\, (\abs{u}_1+\abs{v}_1)
      e^{-\abs{u}_1-\abs{v}_1} \qquad\quad \\
    \times\sum_{\alpha,\beta\in\N^n} \prod_{j=1}^n
    \Bigl(
        f(\lambda_{\alpha_j}^L)e^{(u_j+v_{j-1})t(\lambda_{\alpha_j}^L - E)}
        g(\mu_{\beta_j}^L)e^{-(u_j+v_j)t(\mu_{\beta_j}^L - E)} \\ \qquad
        \times\<\varphi_{\alpha_j}^L, V\psi_{\beta_j}^L\>
        \<\psi_{\beta_j}^L,V\varphi_{\alpha_{j+1}}^L\>\Bigr),
    \end{multline}
    from which the assertion follows.
\end{proof}

\subsection{Smoothing and infinite-volume operators}
\label{ssec:smoothing}

Throughout this subsection, $a\in(0,1)$ and $n\in\N$ are fixed. We
also fix a cut-off energy $E_0\ge 1$ and a Fermi energy
$E\in [-E_0+1,E_0 -1]$.

The goal is to apply Lemma~\ref{lemma:trace_using_feynman} using
suitable functions $f$ and $g$ and to rewrite the right-hand side of
\eqref{eq:Feynman_integrated_trace} as a trace involving the
infinite-volume operators $H$ and $H'$.
Switching from finite-volume to infinite-volume operators constitutes
the core of the argument.
The technical tool to implement this switch to infinite-volume
objects is the Helffer--Sj\"ostrand formula, which supplies the proof
of Lemma~\ref{lemma:hs} below.
Since it is applicable to sufficiently smooth functions only, we define
appropriately smoothed versions of indicator functions.

\begin{definition} \label{def:smooth_cutoff}
  Given a length $L>1$, we say that $\chi_L^\pm \in C_c^\infty(\R)$
  are \emph{smooth cut-off functions at energy $E$}, if they obey
  \begin{equation}
    \begin{split}
      1_{[E + 2L^{-a}, E_0]} & \le \chi_L^+ \le 1_{(E+ L^{-a} , E_0+1)}, \\
      1_{[-E_0,E - 2L^{-a}]} & \le \chi_L^- \le 1_{(-E_0-1,E-L^{-a})}
    \end{split}
  \end{equation}
  and if there exist $L$-independent constants $c_k > 0$
  for $k\in\N_0$, such that
  \begin{equation} \label{def:chi_L2}
    \chi_L^\pm (E \pm L^{-a} \pm x)\le c_0 L^a \,x
  \end{equation}
  for all $x\in[0, L^{-a})$ and
  \begin{equation} \label{def:chi_L3}
    \biggabs{
      \frac{\d^k}{\d x^k} \chi^\pm_L(E \pm L^{-a} \pm x)}
    \le
    \begin{cases}
      c_k L^{ak} & \text{if } 0\le x < L^{-a}, \\
      c_k & \text{otherwise}
    \end{cases}
  \end{equation}
  for every $k\in\N$ and $x\in\R$.
  We choose the smooth decay of $\chi_L^+$ in $[E_0,E_0+1]$
  independently of $L$, and analogously for $\chi_L^-$.
	Clearly such functions exist.
	Fig.~\ref{fig:chiLpm} illustrates the behaviour of $\chi_L^\pm$.
\end{definition}

\begin{figure}
  {\centering  \includegraphics[width=\textwidth]{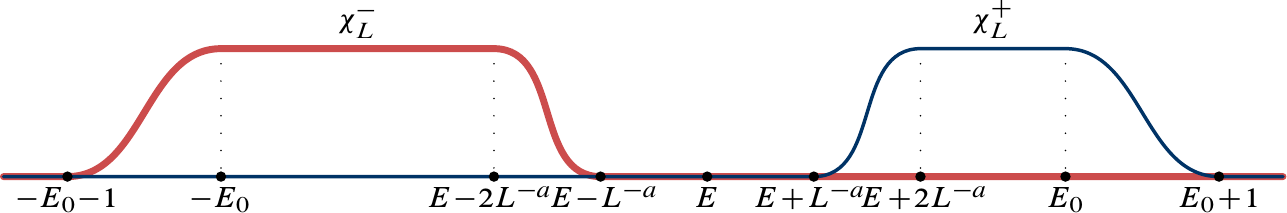}}
  \caption{Sketch of the smooth cut-off functions $\chi_L^\pm$}
  \label{fig:chiLpm}
\end{figure}

We are interested in a lower bound for the left-hand side
of~\eqref{eq:nth-term} which is proportional to $\ln L$ up to
subdominant corrections.

\begin{lemma} \label{lemma:trace_bound_feynman_finite}
  Let $L > 1$. Then
  \begin{align} \label{eq:trace_bound_feynman_finite}
    &\tr\bigbraces{\bigparens{\1_{(-\infty,E]}(H_L)\1_{(E,\infty)}(H_L')}^n} \notag\\
    &\ge
      \int_{(0,\infty)^n\times (0,\infty)^n}\d(u,v)\, (\abs{u}_1+\abs{v}_1)
    e^{-\abs{u}_1-\abs{v}_1}
    \int_0^\infty\dt\, t^{2n-1} \notag \\
    & \;\;\times\,\tr\bigg\{\prod_{j=1}^n
        \sqrt{V}\chi_L^-(H_L) e^{(u_j + v_{j-1})t(H_L-E)}V
                \chi_L^+(H_L') e^{-(u_j + v_j)t(H_L'-E)}\sqrt{V}\bigg\}. 
  \end{align}
\end{lemma}

\begin{proof}
    The inequalities
    \begin{equation} \label{eq:indicator_chi_smaller}
      \1_{(-\infty,E]} \ge \chi_L^-
        \quad \text{and} \quad
        \1_{(E,\infty)} \ge \chi_L^+,
    \end{equation}
    together with the cyclicity of the trace, imply
    \begin{equation}
      \tr\bigbraces{\bigparens{\1_{(-\infty,E]}(H_L)\1_{(E,\infty)}(H_L')}^n}
        \ge
        \tr\bigbraces{\bigparens{\chi_L^-(H_L)\chi_L^+(H_L')}^n}.
    \end{equation}
    Together with Lemma~\ref{lemma:trace_using_feynman}, this yields
    the claim.
\end{proof}

\begin{remark}
  In the sequel, we determine the exact asymptotics of the right-hand
  side of~\eqref{eq:trace_bound_feynman_finite}. Thus it is only the
  smoothing introduced in Lemma~\ref{lemma:trace_bound_feynman_finite}
  which prevents us from determining the exact asymptotics in
  Theorem~\ref{theorem:nth-term}.
\end{remark}

The following technical lemma constitutes the core of the arguments in the
present subsection.

\begin{lemma} \label{lemma:hs}
  For $L > 1$, $t \ge 0$ and $x\in\R$ we define
  \begin{equation} \label{eq:fLt_gLt}
    f_L^t(x)
    :=
    \chi_L^-(x) e^{t(x-E)}
    \quad \text{and} \quad
    g_L^t(x)
    :=
    \chi_L^+(x) e^{-t(x-E)},
  \end{equation}
  and $h_L^t$ will stand for either $f_L^t$ or $g_L^t$.
  Let $\varepsilon\in(0, 1 - a)$ and $M\in\N \setminus\{1\}$.
  Then there are constants $c > 0$, $L_0 > 1$ and
  a polynomial $Q_M$ of degree $M+1$ with non-negative coefficients,
  such that for every $t\ge 0$, every $L \ge L_0$ and every
  $\varepsilon\in(0,1-a)$ the estimate
  \begin{multline}
    \bignorm{\sqrt{V}\bigparens{h_L^t(H_L^{(\prime)}) - h_L^t(H^{(\prime)})}\sqrt{V}}
    \\ \le Q_M(t/L^a) \bigparens{L^{a - M(1 - a - \varepsilon)}
      + L^{d+a(M+1)} e^{-c L^\varepsilon}}
  \end{multline}
  holds.
  \begin{proof}
    The proof essentially follows the ideas of a part of the proof in
    \cite[Lemma~3.14]{GKM}. The idea is to apply the
    Helffer--Sj\"ostrand formula to estimate the difference of
    resolvents, cf. formulae (4.13) and (4.16) in \cite{GKM}. For a
    detailed exposition, see \cite{Heiner}.
  \end{proof}
\end{lemma}

Before we prove the main assertion of this subsection, we need a
spectral-gap estimate.
\begin{lemma} \label{lemma:supmeasurefinite}
  There is a constant $C > 0$ such that for every $L > 1$ and every
  $t\ge 0$ we have
  \begin{equation}
    \tr\bigbraces{\sqrt{V}h_L^t(H^{(\prime)}_{(L)})\sqrt{V}}
    \le C e^{-tL^{-a}},
  \end{equation}
  where $h_L^t\in C_c^\infty(\R)$ is as in Lemma~\ref{lemma:hs}.
  \begin{proof}
    For the first assertion, note that there is a bounded interval
    $I\subseteq\R$ such that $h_L^t\le \1_I e^{-tL^{-a}}$ for all
    $t\ge 0$ and $L > 1$.
    Thus,
    \begin{align}
      \tr\bigbraces{\sqrt{V}h_L^t(H^{(\prime)}_{(L)})\sqrt{V}}
      & \le
      e^{-tL^{-a}}\tr\bigbraces{\sqrt{V}1_I(H^{(\prime)}_{(L)})\sqrt{V}} \notag\\
      & \le e^{-tL^{-a}}e^{\sup
      I}\tr\bigbraces{\sqrt{V}e^{-H_{(L)}^{(\prime)}}\sqrt{V}} \notag\\
      & \le e^{-tL^{-a}}e^{\sup I}\tr\bigbraces{\sqrt{V}e^{-H^{(\prime)}}\sqrt{V}}
    \end{align}
    for $t \ge 0$ and $L > 1$. The last inequality and the finiteness
    of $\tr\bigbraces{\sqrt{V}e^{-H^{(\prime)}}\sqrt{V}}$ follow
    from \cite[Thm.~6.1]{MR1756112}.
  \end{proof}
\end{lemma}

The next lemma accomplishes the transition from finite-volume
to infinite-volume operators.
\begin{lemma} \label{lemma:hs_application}
  For $L > 1$ and $t \ge 0$, let $f_L^t,g_L^t\in C_c^\infty(\R)$ as in
  Lemma~\ref{lemma:hs}. Let $u,v\in(0,\infty)^n$. Then
  \begin{multline} \label{eq:hs_application}
    \int_0^\infty \dt\; t^{2n-1}
    \tr\biggl\lvert
      \prod_{j=1}^n\sqrt{V}f_L^{(u_j + v_{j-1})t}(H_L)Vg_L^{(u_j + v_j)t}(H_L')\sqrt{V}
      \\ \quad -
      \prod_{j=1}^n\sqrt{V}f_L^{(u_j + v_{j-1})t}(H)Vg_L^{(u_j + v_j)t}(H')\sqrt{V}
    \biggr\rvert
    \\ \quad \le \frac{1}{(\abs{u}_1 + \abs{v}_1)^{2n}} \, \oh(1) \hfill
  \end{multline}
  as $L\to\infty$, where the $\oh(1)$-term does not depend on $u$ or
  $v$. We also used the convention $v_0 := v_n$.
\end{lemma}
\begin{proof}
  To shorten formulas, we introduce a vector $\alpha\in(0,\infty)^{2n}$ via
  \begin{equation}\label{eq:alpha-2n}
    \alpha_{2j-1} := u_j + v_{j-1} \quad\text{and}\quad
    \alpha_{2j} := u_j + v_j
  \end{equation}
  for $1\le j\le n$ and operators
  \begin{equation}
    A^{(L)}_k
    :=
    \begin{cases}
      \sqrt{V}f_L^{\alpha_k t}(H_{(L)})\sqrt{V} & \text{for $k$ odd,} \\
      \sqrt{V}g_L^{\alpha_k t}(H_{(L)}')\sqrt{V} & \text{for $k$ even}
    \end{cases}
  \end{equation}
  for $1\le k\le 2n$. The difference of operator products
  in~\eqref{eq:hs_application} is then
  \begin{equation} \label{eq:difference-of-operator-products}
    \prod_{j=1}^{2n} A_j^L - \prod_{j=1}^{2n} A_j
    =
    \sum_{k=1}^{2n} A_1\dotsm A_{k-1}(A^L_k - A_k) A_{k+1}^L\dotsm A^L_{2n}.
  \end{equation}
  The trace norm of this difference can be estimated using
  Lemma~\ref{lemma:supmeasurefinite}: There is a constant $C > 0$ such
  that
  \begin{align}
    \tr\biggabs{
      \prod_{j=1}^{2n} A_j^L - \prod_{j=1}^{2n} A_j
    }
    & \le
    \sum_{k=1}^{2n} \norm{A_k^L - A_k}
    \biggparens{\prod_{j=1}^{k-1}\tr\abs{A_j}}
    \biggparens{\prod_{j=k+1}^{2n}\tr\abs{A_j^L}} \notag\\
    & \le
    C^{2n-1}\sum_{k=1}^{2n} \norm{A_k^L - A_k}
    e^{-(\abs{\alpha}_1-\alpha_k)t L^{-a}},
  \end{align}
  where $\abs{\alpha}_1 = \alpha_1 + \dotsb + \alpha_{2n}$ denotes the
  1-norm of $\alpha\in(0,\infty)^{2n}$. We estimate the $k$th term in
  this sum. Let $\varepsilon\in(0,1 - a)$ and $M\in\N$. For $L$
  sufficiently large, Lemma~\ref{lemma:hs} implies
  \begin{multline} \label{eq:hs-lemma-applied}
    \norm{A_k^L - A_k} e^{-(\abs{\alpha}_1-\alpha_k)t L^{-a}}
    \\\le
    Q_M(\alpha_k t/L^a) \bigparens{L^{a - M(1 - a - \varepsilon)}
      + L^{d+a(M+1)} e^{-c L^\varepsilon}} e^{-(\abs{\alpha}_1-\alpha_k)tL^{-a}},
  \end{multline}
  where $Q_M(x) = \sum_{\ell=0}^{M+1} q_\ell x^\ell$ is the polynomial in
  Lemma~\ref{lemma:hs} with non-negative coefficients $q_\ell$.
  Integrating \eqref{eq:hs-lemma-applied} yields
  \begin{align} \label{eq:integrated-hs-lemma-applied}
    \int_0^\infty \dt\, & t^{2n-1} \norm{A_k^L - A_k}
    e^{-(\abs{\alpha}_1-\alpha_k)t L^{-a}}
    \notag\\
    &\le \bigparens{L^{a - M(1 - a - \varepsilon)}
            + L^{d+a(M+1)} e^{-c L^\varepsilon}} \notag\\
            & \qquad \times\, \sum_{\ell=0}^{M+1} q_\ell \int_0^\infty\!\dt\, t^{2n-1}
          e^{-(\abs{\alpha}_1-\alpha_k)tL^{-a}}
          \frac{\alpha_k^\ell t^\ell}{L^{a\ell}}
    \notag\\
    &=
      \bigparens{L^{a - M(1 - a - \varepsilon)} + L^{d+a(M+1)} e^{-c
        L^\varepsilon}} \notag \\
    & \qquad \times\,     
      \sum_{\ell=0}^{M+1} \frac{q_\ell\Gamma(2n+\ell) \alpha_k^\ell L^{a(2n+\ell)}}%
        {L^{a\ell}(\abs{\alpha}_1-\alpha_k)^{2n+\ell}},
  \end{align}
  where $\Gamma$ denotes Euler's Gamma Function.
  The definition of $\alpha\in(0,\infty)^{2n}$ in \eqref{eq:alpha-2n} yields
  $\abs{\alpha}_1 = 2(\abs{u}_1 + \abs{v}_1)$, and
  thus $\abs{\alpha}_1-\alpha_k\ge \abs{u}_1 +
  \abs{v}_1\ge \alpha_k$. This makes the right-hand side
  of~\eqref{eq:integrated-hs-lemma-applied} smaller than
  \begin{equation}
    C_M \frac{L^{2na}}{(\abs{u}_1 + \abs{v}_1)^{2n}}
    \bigparens{L^{a - M(1 - a - \varepsilon)} + L^{d+a(M+1)} e^{-c L^\varepsilon}}
  \end{equation}
  with some constant $C_M > 0$ depending on $Q_M$ and $n$.
  For given $\varepsilon < 1 - a$, we can choose
  $M$ large enough for the $L$-terms to vanish as $L\to\infty$.
\end{proof}

Using Lemma~\ref{lemma:hs_application}, we can rewrite the right-hand
side of \eqref{eq:trace_bound_feynman_finite}.
\begin{corollary} \label{corollary:trace_bound_feynman_infinite}
  The estimate
  \begin{align} \label{eq:trace_bound_feynman_infinite}
  \tr & \bigbraces{\bigparens{ \1_{(-\infty,E]}(H_L)\1_{(E,\infty)}(H_L')}^n} \notag\\
    &
    \ge
    \int_{(0,\infty)^n\times (0,\infty)^n}\d(u,v)\, (\abs{u}_1+\abs{v}_1)
      e^{-\abs{u}_1-\abs{v}_1}
    \int_0^\infty\dt\, t^{2n-1} \notag\\
    &
    \qquad\times \tr\bigg\{\prod_{j=1}^n
        \sqrt{V}\chi_L^-(H) e^{(u_j + v_{j-1})t(H-E)}V
                \chi_L^+(H') e^{-(u_j + v_j)t(H'-E)}\sqrt{V}\bigg\} \notag\\
    & \quad + \oh(1)
  \end{align}
  holds as $L\to\infty$.
\end{corollary}

\begin{proof}
    The claim follows from Lemma~\ref{lemma:trace_bound_feynman_finite}
    and Lemma~\ref{lemma:hs_application},
    which imply that the integral
    \begin{multline} \label{eq:finite_infinite_comparison}
      \int_{(0,\infty)^n\times (0,\infty)^n}\d(u,v)\, (\abs{u}_1+\abs{v}_1)
      e^{-\abs{u}_1-\abs{v}_1}
      \int_0^\infty\dt\, t^{2n-1}\\
      \quad\times\Biggl(
        \tr\bigg\{\prod_{j=1}^n\sqrt{V}f_L^{(u_j + v_{j-1})t}(H_L)Vg_L^{(u_j
        + v_j)t}(H_L')\sqrt{V}\bigg\} \\ 
       - \tr\bigg\{\prod_{j=1}^n\sqrt{V}f_L^{(u_j + v_{j-1})t}(H)Vg_L^{(u_j
        + v_j)t}(H')\sqrt{V} \bigg\}
      \Biggr)
    \end{multline}
    vanishes in the limit $L\to\infty$, because
    \begin{equation}
      \int_{(0,\infty)^n\times (0,\infty)^n}\d(u,v)\,
      \frac{e^{-\abs{u}_1-\abs{v}_1}}{(\abs{u}_1+\abs{v}_1)^{2n-1}}
      = \frac{1}{(2n-1)!}
      < \infty,
    \end{equation}
    as can be seen from the coarea formula.
\end{proof}

\begin{remark}
  Comparing the smooth cut-off functions $\chi_L^\pm$
  with the ones in \cite[Def.~3.13]{GKM}, the difference is
  that the
  cut-off functions there have $E$ as the boundary of their support,
  while the ones here have distance $L^{-a}$ between $E$ and their
  support. To compensate for this, the $t$-integral has been cut off at $t =
  L^{-a}$ in \cite[Lemma~3.11]{GKM}, which yields a lower bound for $n
  = 1$. For $n\ge 2$, it is not immediately clear if the integrand
  in~\eqref{eq:trace_bound_feynman_infinite} is positive, so cutting
  off the integration might not result in a lower bound; this is the
  reason for choosing the cut-off functions different from those in
  \cite{GKM}.
\end{remark}

\subsection{Infinite-volume trace expressions}
\label{ssec:undoing-the-smoothing}

Throughout this subsection, we fix $a\in(0,1)$, \ $n\in\N$ and a
cut-off energy $E_0\ge 1$.

In Corollary~\ref{corollary:trace_bound_feynman_infinite},
we gave a lower bound on the $n$th term of
\eqref{eq:log_overlap_series} in which only infinite-volume operators
occur. In order to control the errors in that step, it was necessary
to introduce smoothed versions of indicator functions in
\eqref{eq:indicator_chi_smaller}. In the present subsection, our aim is
to replace these smoothed functions with discontinuous ones,
which will allow us to determine the asymptotics of the resulting
expression.

We introduce measures
$\mu^1,\nu^1\from\Borel(\R)\to[0,\infty]$ defined by
\begin{align} \label{eq:1D_measures}
  \mu^1(A) & := \tr\big\{\sqrt{V}\1_A(H)\sqrt{V}\big\},
  &&&
  \nu^1(B) & := \tr\big\{\sqrt{V}\1_B(H')\sqrt{V}\big\}
\end{align}
for $A,B\in\Borel(\R)$.  The expressions in \eqref{eq:1D_measures} are
finite for bounded Borel sets as a consequence of
\cite[Thm.~B.9.2]{MR670130}.

The absolutely continuous parts of the measures $\mu^1$ and $\nu^1$
will turn out to be important. To define their densities in an
applicable manner, we use a limiting absorption principle due to
Birman and {\`E}ntina.
\begin{proposition}[\protect{\cite[Lemma 4.3]{BiEn67e}}]
  \label{proposition:Birman}
  There exists a null set $\Nn_0\subset\R$ such that the limits
  \begin{align} \label{eq:BirmanA}
    A(E) & :=
    \lim_{\varepsilon\downto 0}
    \frac{1}{2\varepsilon}\sqrt{V}\1_{(E-\varepsilon,E+\varepsilon)}(H)\sqrt{V},
    \\ \label{eq:BirmanB}
    B(E) & :=
    \lim_{\varepsilon\downto 0}
    \frac{1}{2\varepsilon}\sqrt{V} \1_{(E-\varepsilon,E+\varepsilon)}(H')\sqrt{V}
  \end{align}
  exist in trace class for all $E\in\R\setminus\Nn_0$ and define
  non-negative trace class operators $A(E)$ and $B(E)$.
\end{proposition}

In the next lemma we identify the densities of the absolutely
continuous parts of $\mu^1$ and $\nu^1$. The proof of this lemma
follows directly from the definitions.

\begin{lemma} \label{lemma:densities}
  The functions $E\mapsto\tr A(E)$, respectively $E\mapsto\tr B(E)$,
  are locally integrable Lebesgue densities of the
  absolutely continuous parts of $\mu^1$, respectively $\nu^1$.
\end{lemma}

We will need an auxiliary statement for the main result of this
subsection.

\begin{lemma} \label{lemma:exp_leb_point}
  Let $\mu$ be a locally finite Borel measure on $\R$. Let $c_0 > 0$
  and $0 < \varepsilon < \delta < c_0$. Then for a.e.\ $x_0\in\R$ there
  is a constant $C$, depending on $x_0$, $c_0$ and $\mu$, such
  that for all $t > 0$
  \begin{equation}
    \int_{[x_0,x_0+\delta]}\d\mu(x)\,e^{-t(x-x_0)}
    \le C\frac{1 - e^{-t\delta}}{t}
  \end{equation}
  and
  \begin{align}
    \int_{[x_0+\varepsilon,x_0+\delta]}\d\mu(x)\,e^{-t(x-x_0)}
    & \le
    C e^{-t\varepsilon/2}\frac{1 - e^{-t\delta/2}}{t/2}
    \notag\\ &\le
    C \frac{e^{-t\varepsilon/2}}{t/2}.
  \end{align}
  The exceptional set of values of $x_0$ for which the assertion does
  not hold depends neither on $c_0$, $\varepsilon$ nor $\delta$.
  \begin{proof}
    The constant 
    \begin{equation}
      C
      :=
      \sup_{\eta\in(0,c_0)}\frac{1}{\eta}\mu([x_0,x_0+\eta])
    \end{equation}
    is finite for a.e.\ $x_0\in\R$.
    We compute using Tonelli's theorem
    \begin{align}
      \int_{[x_0,x_0+\delta]} &\d\mu(x)\, e^{-t(x-x_0)}
      \notag\\
      & =
      \int_{[x_0,x_0+\delta]}\d\mu(x)\Bigparens{e^{-t\delta} +
      t\int_x^{x_0+\delta}\d\xi\, e^{-t(\xi-x_0)}} \notag\\
      & =
      \delta e^{-t\delta} \frac{1}{\delta}\mu([x_0,x_0+\delta])
      + t\int_{x_0}^{x_0+\delta}\d\xi\int_{[x_0,\xi]}\d\mu(x)\,
        e^{-t(\xi-x_0)} \notag \\
      & \le C\delta e^{-t\delta} +
      t\int_{x_0}^{x_0+\delta}\d\xi\, e^{-t(\xi-x_0)} \frac{\xi-x_0}{\xi-x_0}\mu([x_0,\xi])
      \notag \\
      & \le C\delta e^{-t\delta} + C
      t\int_0^\delta \d\xi\, \xi e^{-t\xi}
      \notag\\
      & =
      C \frac{1 - e^{-t\delta}}{t}.
    \end{align}
    The second assertion follows from the first one and the bound
    $e^{-t(x-x_0)}\le e^{-t\varepsilon/2}e^{-t(x-x_0)/2}$ for
    $\varepsilon\le x - x_0\le \delta$.
  \end{proof}
\end{lemma}

\begin{definition} \label{definition:u-integral-finite}
  \begin{nummer}
  \item
    For $k\in\N$, we define
      \begin{equation} \label{eq:u-integral-finite}
        I_k := \int_{(0,\infty)^k}\du\,
        \frac{\abs{u}_1e^{-\abs{u}_1}}{\prod_{j=1}^k (u_j + u_{j+1})},
      \end{equation}
      where $u_{k+1} := u_1$ for $u\in\R^k$.
  \item
    We define discontinuous $L$-independent functions
    $\chi^\pm\from\R\to[0,1]$ by
      \begin{equation}
        \chi^- := \max\set{\chi^-_L, \1_{[-E_0,E)}}
        \quad
        \text{and}
        \quad
        \chi^+ := \max\set{\chi^+_L, \1_{(E, E_0]}}.
      \end{equation}
  \end{nummer}
\end{definition}

\begin{remarks} \label{remarks:u-integral-finite}
  \item
    The integral $I_k$ will be discussed further in
  	Subsection~\ref{ssec:u-integral}; in particular, $I_k$ is finite
  	for every $k\in\N$.
  \item 
  	The functions $\chi_L^\pm$ converge pointwise to $\chi^\pm$ as
  	$L\to\infty$.
  	They are obtained from replacing the smooth $L$-dependent part by a
  	discontinuous step at $E$. Fig.~\ref{fig:chiplus} illustrates the
  	behaviour of $\chi^\pm$.
\end{remarks}

\begin{figure}
  {\centering \includegraphics[width=\textwidth]{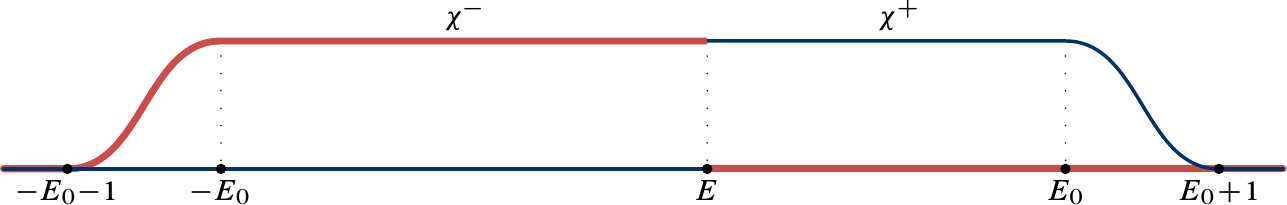}}
  \caption{Sketch of the discontinuous cut-off functions $\chi^\pm$}
  \label{fig:chiplus}
\end{figure}

The following lemma is the main result of the current section.

\begin{lemma} \label{lemma:error2_and_3}
  There is a null set $\Nn\subseteq\R$ which does not depend on $a$, $n$
  and $E_0$,
  such that for every $E\in[-E_0,E_0]\setminus\Nn$,
  \begin{align} \label{eq:chiL_vs_chi}
    \int_0^\infty\dt\, t^{2n-1}
    &\int_{(0,\infty)^n\times 
      (0,\infty)^n}\d(u,v)\, (\abs{u}_1+\abs{v}_1) e^{-\abs{u}_1-\abs{v}_1}
    \notag\\
    & \quad\;\times \!\Biggl[
        \tr\bigg\{
          \prod_{j=1}^n\!
         \sqrt{V} \chi^-_L(H)  e^{(u_j + v_{j-1})t(H-E)} \sqrt{V} \notag\\
        & \hspace*{3.3cm} \times \;      \sqrt{V} \chi^+_L(H') e^{-(u_j + v_j)t(H'-E)} \sqrt{V}\bigg\}
            \notag \\
        & \qquad\quad - e^{-tL^{-a}} \tr\bigg\{
          \prod_{j=1}^n\!
            \sqrt{V} \chi^-(H)  e^{(u_j + v_{j-1})t(H-E)}\sqrt{V} \notag\\
        & \hspace*{3.3cm} \times \;\sqrt{V}\chi^+(H') e^{-(u_j + v_j)t(H'-E)} \sqrt{V}\bigg\}
      \Biggr]
      \notag \\
      & \quad = \Oh(1)
  \end{align}
  as $L\to\infty$, where the $\Oh(1)$-term depends on $a$, $n$, $E$
  and $E_0$.
  \begin{proof}
    First, notice that if $f_j,g_j$ are bounded measurable functions of
    compact support for $1\le j\le n$, then
    \begin{align} \label{eq:product-of-traces}
      \tr \prod_{j=1}^n
        \biggabs{\sqrt{V} f_j(H) V g_j(H') \sqrt{V}}
      & \le \prod_{j=1}^n \tr\Bigbraces{\sqrt{V}f_j(H)\sqrt{V}}
                       \tr\Bigbraces{\sqrt{V}g_j(H')\sqrt{V}} \notag\\
      & =
      \int_{\R^n}\d\mu^n(x)\int_{\R^n}\d\nu^n(y) \prod_{j=1}^n f_j(x_j)g_j(y_j),
    \end{align}
    where we wrote $\mu^n$ and $\nu^n$ for the $n$-fold product
    measure of $\mu^1$ and $\nu^1$, respectively.

    For brevity, let $\delta := L^{-a}$. We introduce a vector
    $\alpha\in(0,\infty)^{2n}$ via
    \begin{equation}
      \alpha_{2j-1} := u_j + v_{j-1} \quad\text{and}\quad
      \alpha_{2j} := u_j + v_j
    \end{equation}
    for $1\le j\le n$ and operators
    \begin{equation}
      A^{(L)}_k
      :=
      \begin{cases}
        \sqrt{V}\chi^-_{(L)}(H)e^{\alpha_kt(H-E)}\sqrt{V} & \text{for $k$ odd,} \\
        \sqrt{V}\chi^+_{(L)}(H')e^{-\alpha_kt(H'-E)}\sqrt{V} & \text{for $k$ even}
      \end{cases}
    \end{equation}
    for $1\le k\le 2n$. The difference of operator
    products in~\eqref{eq:chiL_vs_chi}  then equals
    \begin{equation} \label{eq:error2-and-3-split}
      \prod_{j=1}^{2n} A_j^L - e^{-t\delta}\prod_{j=1}^{2n} A_j
      =
      e^{-t\delta}\Bigparens{\prod_{j=1}^{2n} A_j^L -
      \prod_{j=1}^{2n} A_j}
      + (1-e^{-t\delta})\prod_{j=1}^{2n} A_j^L,
    \end{equation}
    where as in~\eqref{eq:difference-of-operator-products},
    \begin{equation} \label{eq:difference-of-operator-products-again}
      \prod_{j=1}^{2n} A_j^L - \prod_{j=1}^{2n} A_j
      =
      \sum_{k=1}^{2n} A_1\dotsm A_{k-1}(A^L_k - A_k) A_{k+1}^L\dotsm A^L_{2n}.
    \end{equation}

    We will treat the two terms on the right-hand side
    of~\eqref{eq:error2-and-3-split} individually. For the first term,
    we estimate the $k$th term
    in~\eqref{eq:difference-of-operator-products-again}. We will carry
    out the argument in the case where $k$ is even. The argument is
    similar for odd $k$. Since $0\le\chi^+ - \chi^+_L \le
    \1_{[E,E+2\delta]}$, \ $\chi_{(L)}^-\le\1_{[-E_0-1,E]}$ and
    $\chi_{(L)}^+\le\1_{[E,E_0+1]}$,
    \eqref{eq:product-of-traces} implies
    \begin{align} \label{eq:exp_leb_point-applied}
      &\tr\Bigabs{A_1\dotsm A_{k-1}(A^L_k - A_k) A_{k+1}^L\dotsm
      A^L_{2n}}
      \notag\\
      &\quad\le
      \int_{[-E_0-1,E]^n}\d\mu^n(x)\int_{[E,E_0+1]^n}\d\nu^n(y)\,
      \1_{[E,E+2\delta]}(y_k)
      \notag\\
      &\qquad\qquad\times\exp\Bigparens{-t\sum_{j=1}^n\bigparens{(u_j+v_j)(y_j - E) -
      (u_j + v_{j-1})(x_j-E)}} \notag\\
      &\quad\le
      C
      \frac{1 - e^{-2(u_k+v_k)t\delta}}{t}
      \frac{1}{t^{2n-1}}
      \frac{1}{\prod_{j=1}^n (u_j+v_j)(u_j+v_{j-1})},
    \end{align}
    where $C$ is some finite constant and the last inequality follows
    for a.e.\ $E\in [-E_0,E_0]$ from applying
    Lemma~\ref{lemma:exp_leb_point} to every integral and the estimate
    $1 - e^{-tx}\le 1$ to all but the $k$th term. Using the bound $1 -
    e^{-2(u_k+v_k)t\delta}\le 2(u_k+v_k)t\delta$, we conclude
    \begin{equation}
      e^{-t\delta}\tr\biggabs{\prod_{j=1}^{2n} A_j^L -
      \prod_{j=1}^{2n} A_j}
      \le 4C \frac{\delta e^{-t\delta}}{t^{2n-1}}
      \frac{\abs{u}_1 + \abs{v}_1}{\prod_{j=1}^n (u_j+v_j)(u_j+v_{j-1})},
    \end{equation}
    and therefore
    \begin{align}
      \int_0^\infty & \dt\, t^{2n-1}
      \int_{(0,\infty)^n\times
        (0,\infty)^n}\d(u,v)\, (\abs{u}_1+\abs{v}_1)
      e^{-\abs{u}_1-\abs{v}_1} \notag\\
      &\qquad\times e^{-t\delta}\tr\biggabs{\prod_{j=1}^{2n} A_j^L -
      \prod_{j=1}^{2n} A_j}
      \notag\\
      &\;\le
      4C
      \int_0^\infty\dt\, \delta e^{-t\delta}
      \int_{(0,\infty)^n\times (0,\infty)^n}\d(u,v)\,
      \frac{(\abs{u}_1+\abs{v}_1)^2 e^{-\abs{u}_1-\abs{v}_1}}{\prod_{j=1}^n (u_j+v_j)(u_j+v_{j-1})}.
    \end{align}
    Here, the $t$-integral yields $1$ for every $\delta > 0$, and
    the $(u,v)$-integral is finite since
    \begin{multline}
      \int_{(0,\infty)^n\times (0,\infty)^n}\d(u,v)\,
      \frac{(\abs{u}_1+\abs{v}_1)^2 e^{-\abs{u}_1-\abs{v}_1}}{\prod_{j=1}^n (u_j+v_j)(u_j+v_{j-1})}
      \\ \qquad\qquad \le
      \int_{(0,\infty)^n\times (0,\infty)^n}\d(u,v)\,
      \frac{(\abs{u}_1+\abs{v}_1) e^{-\abs{u}_1/2-\abs{v}_1/2}}{\prod_{j=1}^n (u_j+v_j)(u_j+v_{j-1})} \hfill \\
       \qquad\qquad  = 2 I_{2n} < \infty, \hfill
    \end{multline}
    with $I_{2n}$ as in Definition~\ref{definition:u-integral-finite}.
    This shows that the integral of the trace norm of the first term
    on the right-hand side of~\eqref{eq:error2-and-3-split} yields an
    error that remains finite as $L\to\infty$.

    The trace norm of the second term on the right-hand side
    of~\eqref{eq:error2-and-3-split} is
    \begin{align} \label{eq:exp_leb_point-applied-second-time}
      &(1-e^{-t\delta})\tr\biggabs{\prod_{j=1}^{2n} A_j^L}
      \notag\\ & \le
        (1 - e^{-t\delta})
        \int_{[-E_0-1,E-\delta]^n}\d\mu^n(x)\int_{[E+\delta,E_0+1]^n}\d\nu^n(y)\,
        \notag\\
      & \qquad \times\exp\Bigparens{-t\sum_{j=1}^n\bigparens{(u_j+v_j)(y_j - E) -
      (u_j + v_{j-1})(x_j-E)}}
      \notag \\
      & \le
        C
        (1 - e^{-t\delta})
        \prod_{j=1}^n
        \frac{e^{-(u_j+v_{j-1})t\delta/2}e^{-(u_j+v_j)t\delta/2}}%
             {(u_j+v_{j-1})(u_j+v_j)(t/2)^2}
      \notag \\
      & =
      2^{2n}C
      \frac{(1 - e^{-t\delta}) e^{-(\abs{u}_1+\abs{v}_1)t\delta} t^{-2n}}{\prod_{j=1}^n (u_j+v_j)(u_j+v_{j-1})},
    \end{align}
    where the first inequality is a consequence of
    $\chi_{(L)}^-\le\1_{[-E_0-1,E]}$ and
    $\chi_{(L)}^+\le\1_{[E,E_0+1]}$, and the second inequality follows
    for a.e.\ $E\in[-E_0,E_0]$ from
    Lemma~\ref{lemma:exp_leb_point}. Now, we perform the $t$- and
    $(u,v)$-integration
    \begin{align}
      &\int_{(0,\infty)^n\times (0,\infty)^n}\d(u,v)\, (\abs{u}_1+\abs{v}_1)
      e^{-\abs{u}_1-\abs{v}_1} \int_0^\infty\dt\, t^{2n-1}
      \notag\\
      & \hspace{8em} \times
      \frac{(1 - e^{-t\delta}) e^{-(\abs{u}_1+\abs{v}_1)t\delta} t^{-2n}}{\prod_{j=1}^n (u_j+v_j)(u_j+v_{j-1})},
      \notag\\
      &=
      \int_{(0,\infty)^n\times (0,\infty)^n}\d(u,v)
      \int_0^\infty\dt\, \frac{(1 - e^{-t})
      e^{-(1+t)(\abs{u}_1+\abs{v}_1)}
      (\abs{u}_1+\abs{v}_1)}%
      {t\,\prod_{j=1}^n (u_j+v_j)(u_j+v_{j-1})} \notag\\
      & = I_{2n} \int_0^\infty\dt\, \frac{1 - e^{-t}}{t^2+t} < \infty,
    \end{align}
    where we performed the successive changes of variables
    $t\delta\rightsquigarrow t$ and $(1+t)(u,v)\rightsquigarrow (u,v)$.
  \end{proof}
\end{lemma}

\begin{corollary} \label{corollary:trace_bound_feynman_exponential}
  For a.e.\ $E\in [-E_0, E_0]$, we have
  \begin{align} \label{eq:trace_bound_feynman_exponential}
    &\tr\bigbraces{\bigparens{\1_{(-\infty,E]}(H_L)\1_{(E,\infty)}(H_L')}^n} \notag\\
    &\ge
    \int_{(0,\infty)^n\times (0,\infty)^n}\d(u,v)\, (\abs{u}_1+\abs{v}_1)
      e^{-\abs{u}_1-\abs{v}_1}
    \int_0^\infty\dt\, t^{2n-1} e^{-tL^{-a}} \notag \\
    &\qquad\times\tr\bigg\{\prod_{j=1}^n
        \sqrt{V}\chi^-(H) e^{(u_j + v_{j-1})t(H-E)}V
                \chi^+(H') e^{-(u_j + v_j)t(H'-E)}\sqrt{V}\bigg\} \notag\\
    &\quad  + \Oh(1)
  \end{align}
  as $L\to\infty$. The null set of exceptional energies does not
  depend on $a$, $n$ and $E_0$.
  \begin{proof}
    We combine Corollary~\ref{corollary:trace_bound_feynman_infinite} and
    Lemma~\ref{lemma:error2_and_3}.
  \end{proof}
\end{corollary}

\subsection{The logarithmic divergence}
\label{ssec:logarithmic-divergence}

Throughout this subsection, we fix $a\in(0,1)$, $n\in\N$ and $E_0 \ge
1$.

The goal is to determine the asymptotics of the
right-hand side of~\eqref{eq:trace_bound_feynman_exponential}.

\begin{lemma} 
  \label{lemma:exptraceconv}
  For a.e.\ $E\in[-E_0,E_0]$,
  \begin{equation} \label{eq:AtE_BtE_convergence}
    \begin{aligned}
      A_t(E) & :=
      \sqrt{V} t e^{t(H-E)}\chi^-(H)\sqrt{V}
      \to A(E), \\
      B_t(E) & :=
      \sqrt{V} t e^{-t(H'-E)}\chi^+(H')\sqrt{V}
      \to B(E)
    \end{aligned}
  \end{equation}
  as $t\to\infty$,
  where the convergences are in trace class. Moreover,
  \begin{equation} \label{eq:sup_norm_texp}
    \begin{aligned}
      \sup_{t\ge 0} \norm{A_t(E)}
      & \le
      \sup_{t\ge 0} \tr A_t(E)
      < \infty, \\
      \sup_{t\ge 0} \norm{B_t(E)}
      & \le
      \sup_{t\ge 0} \tr B_t(E)
      < \infty.
    \end{aligned}
  \end{equation}
  \begin{proof}
    We follow \cite[Lemma 3.16]{GKM} and treat the operator $B_t(E)$;
    the assertions for $A_t(E)$ can be proved using analogous
    arguments. Recall that $B_t(E)$ is
    non-negative. For~\eqref{eq:AtE_BtE_convergence}, we show (1)
    convergence of the trace norms and (2) weak convergence of the
    operators. Together, this implies
    convergence in trace class via \cite[Addendum H]{MR2154153}.

    For the trace norms, we compute
    \begin{align} \label{eq:trBtE}
      \tr B_t(E)
      & = \tr\bigbraces{\sqrt{V}te^{-t(H'-E)}\chi^+(H')\sqrt{V}}
      \notag\\ & =
      \int_{[E,E_0]}\d\nu^1(y)\, t e^{-t(y-E)}
      + \int_{[E_0,E_0+1]}\d\nu^1(y)\, \chi^+(y) t e^{-t(y-E)}
    \end{align}
    where the second term converges to zero as $t\to\infty$ for $E <
    E_0$. The first term can be written as
    \begin{equation} \label{eq:delta_convolution}
      \int_{[E,E_0]}\d\nu^1(y)\, t e^{-t(y-E)}
      =
      \parens{\nu^1_{E_0} * \varrho_t}(E),
    \end{equation}
    where we introduced the restricted (finite) measure
    $\nu_{E_0}^1(M) := \nu^1(M\cap [-E_0,E_0])$ for $M\in\Borel(\R)$
    and the approximation of the identity
    $x\mapsto\varrho_t(x) := t e^{tx} \1_{(-\infty,0)}(x)$.
    As $t\to\infty$, the convolution in~\eqref{eq:delta_convolution}
    converges for a.e.\ $E\in[-E_0,E_0]$
    to $\frac{\d\nu_\ac^1}{\d E} = \tr B(E)$,
    see e.g.\ \cite[Subsec. 2.4.1]{MR3052498}.
    Thus, the trace norm of
    $B_t(E)$ converges to that of $B(E)$ as $t\to\infty$.
    This, together with the continuity of $[0,\infty)\ni t\mapsto
    \tr B_t(E)$, which can be seen from~\eqref{eq:trBtE}, implies
    \eqref{eq:sup_norm_texp}.

    For the weak convergence, take some dense countable set
    $\mathcal{D}\subseteq L^2(\R^d)$.
    Then by a similar delta-argument as above,
    \begin{equation}
      \lim_{t\to\infty} \<\varphi, B_t(E)\psi\> = \<\varphi, B(E)\psi\>
    \end{equation}
    for all $\varphi, \psi\in\mathcal{D}$ and all $E\in
    [-E_0,E_0]$ outside a null set depending on $\mathcal{D}$.
    Together with \eqref{eq:sup_norm_texp}, this proves weak
    convergence to $B(E)$ for a.e.\ $E\in [-E_0,E_0]$, see
    \cite[Thm.~4.26]{MR566954}.
  \end{proof}
\end{lemma}

The following quantity will enter the asymptotics we set out to
prove.
\begin{definition} \label{definition:density-diagonal}
  For $E\in\R\setminus\Nn_0$, let
  \begin{equation} \label{eq:density-diagonal}
    \eta_{2n}(E)
    :=
    \tr\bigbraces{\bigparens{A(E)B(E)}^n}
  \end{equation}
  and extend it trivially to a function $\eta_{2n}\from\R\to[0,\infty)$.
  The non-negativity of~\eqref{eq:density-diagonal} can be seen from
  the cyclicity of the trace.
\end{definition}

The next corollary will show that the trace expression
on the right-hand side of~\eqref{eq:trace_bound_feynman_exponential},
times an appropriate power of~$t$, converges to $\eta_{2n}(E)$ in the
$t\to\infty$ limit.

\begin{corollary} \label{corollary:pointwise_conv}
  Let $\alpha_1,\dotsc,\alpha_n,\beta_1,\dotsc,\beta_n > 0$. Then for
  a.e.\ $E\in[-E_0,E_0]$
  \begin{equation}
    t^{2n} \tr\bigg\{\prod_{j=1}^n
    \sqrt{V}\chi^-(H) \alpha_j e^{\alpha_jt(H-E)}V
    \chi^+(H') \beta_j e^{-\beta_jt(H'-E)}\sqrt{V}\bigg\}
    \to
    \eta_{2n}(E)
  \end{equation}
  as $t\to\infty$.
  \begin{proof}
    Using the notation of Lemma~\ref{lemma:exptraceconv}, we have to
    show
    \begin{equation} \label{eq:pointwise_conv_products}
      \biggabs{\tr\bigg\{\prod_{j=1}^n A_{\alpha_j t}(E) B_{\beta_j t}(E)\bigg\}
        - \tr\bigbraces{\bigparens{A(E)B(E)}^n}}
      \to 0
    \end{equation}
    as $t\to\infty$.
    By Lemma~\ref{lemma:exptraceconv}, $\tr\abs{A_{\alpha_jt}(E) -
    A(E)}\to 0$ and $\tr\abs{B_{\beta_jt}(E) - B(E)}\to 0$ as
    $t\to\infty$, while $\sup_{t\ge 0}\norm{A_t(E)}$ and $\sup_{t\ge
    0}\norm{B_t(E)}$ are finite. Writing the difference of operator
    products in~\eqref{eq:pointwise_conv_products} as
    in~\eqref{eq:difference-of-operator-products}, this proves
    the corollary.
  \end{proof}
\end{corollary}

\begin{lemma} \label{lemma:final-value-log}
  Let $f\in L^1_\loc(\R)$ and suppose $\lim_{t\to\infty}f(t)$
  exists. Then
  \begin{equation}
    \lim_{t\to\infty} f(t)
    =
    - \lim_{s\downto 0}
    \frac{1}{\ln s}\int_1^\infty\dt\, t^{-1} e^{-st} f(t).
  \end{equation}
  \begin{proof}
    Take a compact interval $[s_0,c]\subseteq (0,\infty)$. Then
    \begin{equation} \label{eq:dbyds-lebesgue}
      \dbyd{s}\int_1^\infty\dt\, t^{-1} e^{-st} f(t)
      =
      - \int_1^\infty\dt\, e^{-st} f(t)
    \end{equation}
    for $s\in [s_0,c]$, because $\bigabs{\dbyd{s} t^{-1} e^{-st}
    f(t)}\le e^{-s_0t}\abs{f(t)}$, which is integrable on
    $[1,\infty)$. Therefore \eqref{eq:dbyds-lebesgue} holds for all
    $s > 0$. If $\lim_{s\downto 0} \int_1^\infty \dt\,
    t^{-1}e^{-st}f(t)$ exists, then $\lim_{t\to\infty} f(t) = 0$ and
    the assertion holds. Otherwise,
    \begin{align}
      - \lim_{s\downto 0} \frac{1}{\ln s}\int_1^\infty\dt\, t^{-1} e^{-st} f(t)
      & =
      \lim_{s\downto 0} \frac{1}{1/s} \int_1^\infty\dt\, e^{-st} f(t) \notag\\
      & = \lim_{s\downto 0} s \int_0^\infty\dt\, e^{-st} f(t) \notag\\
      & = \lim_{t\to\infty} f(t),
    \end{align}
    where the last equality is the statement of the classical final-value 
    theorem, see \cite[Thm.~34.3]{MR0344810}.
  \end{proof}
\end{lemma}

We are now ready to compute the asymptotics
of the right-hand side of~\eqref{eq:trace_bound_feynman_exponential}
\begin{theorem} \label{theorem:the-asymptotics}
  For a.e.\ $E\in[-E_0,E_0]$,
  \begin{align} \label{eq:the-asymptotics}
    &\lim_{L\to\infty}\frac{1}{a\ln L}
    \int_0^\infty\dt\, t^{2n-1} e^{-tL^{-a}} \!
    \int_{(0,\infty)^n\times (0,\infty)^n} \!\d(u,v)\, (\abs{u}_1+\abs{v}_1)
      e^{-\abs{u}_1-\abs{v}_1}
    \notag\\
    &\quad
    \times\tr\bigg\{\prod_{j=1}^n
      \sqrt{V}\chi^-(H) e^{(u_j + v_{j-1})t(H-E)}V
      \chi^+(H') e^{-(u_j + v_j)t(H'-E)}\sqrt{V} \bigg\} \notag\\
    &= I_{2n} \eta_{2n}(E).
  \end{align}
\end{theorem}

\begin{proof}
    Let $u,v\in (0,\infty)^n$ and define
    \begin{equation}
      Z(u,v) := \prod_{j=1}^n (u_j+v_{j-1})(u_j+v_j).
    \end{equation}
    Using the notation of Lemma~\ref{lemma:exptraceconv}, we see that 
    \begin{align}
      Z(u,v) & t^{2n}\tr\bigg\{\prod_{j=1}^n
      \sqrt{V}\chi^-(H) e^{(u_j + v_{j-1})t(H-E)}\sqrt{V} \notag\\
      & \hspace*{1.5cm}\times \sqrt{V} \chi^+(H') e^{-(u_j + v_j)t(H'-E)}\sqrt{V} \bigg\} \notag\\
      &  = \tr\bigg\{\prod_{j=1}^n A_{(u_j+v_{j-1})t}(E) B_{(u_j+v_j)t}(E)\bigg\}, 
    \end{align}
    where
    \begin{equation} \label{eq:uv_times_t_times_trace_bounded}
      \biggabs{\tr\bigg\{\prod_{j=1}^n A_{(u_j+v_{j-1})t}(E)
      B_{(u_j+v_j)t}(E)\bigg\}}
      \le
      \bigparens{\sup_{t\ge 0} \tr A_t(E) \sup_{t\ge 0} \tr B_t(E)}^n
      < \infty.
    \end{equation}
    By Corollary~\ref{corollary:pointwise_conv},
    \begin{equation} \label{eq:uv_times_t_times_trace_convergent}
      \lim_{t\to\infty}\tr\bigg\{\prod_{j=1}^n A_{(u_j+v_{j-1})t}(E)
      B_{(u_j+v_j)t}(E)\bigg\}
      = \eta_{2n}(E)
    \end{equation}
    for all $u,v\in (0,\infty)^n$. By
    Remark~\ref{remarks:u-integral-finite}~(i),
    \begin{equation} \label{eq:u-integral-finite-again}
      I_{2n} =
      \int_{(0,\infty)^n\times (0,\infty)^n}\d(u,v)\,
      \frac{(\abs{u}_1+\abs{v}_1)
      e^{-\abs{u}_1-\abs{v}_1}}{Z(u,v)} < \infty.
    \end{equation}
    Equations \eqref{eq:uv_times_t_times_trace_bounded}
    to~\eqref{eq:u-integral-finite-again} supply the assumptions of
    the dominated convergence theorem. It yields the convergence
    \begin{equation}
      \lim_{t\to\infty} f(t) = I_{2n}\eta_{2n}(E)
    \end{equation}
    for
    \begin{align}
      f(t) := {} &
      \int_{(0,\infty)^n\times (0,\infty)^n}\d(u,v)\, (\abs{u}_1+\abs{v}_1)
      e^{-\abs{u}_1-\abs{v}_1} \; t^{2n} \notag\\
      &
      \times \tr\bigg\{\prod_{j=1}^n
      \sqrt{V}\chi^-(H) e^{(u_j + v_{j-1})t(H-E)} \sqrt{V} \notag\\
      & \hspace*{1.6cm} \times \sqrt{V}\chi^+(H') e^{-(u_j + v_j)t(H'-E)}\sqrt{V}\bigg\}
      \notag\\
      = {} &
      \int_{(0,\infty)^n\times (0,\infty)^n}\d(u,v)\,
        \frac{(\abs{u}_1+\abs{v}_1)
          e^{-\abs{u}_1-\abs{v}_1}}{Z(u,v)}
      \notag\\
      & \times\tr\bigg\{\prod_{j=1}^n A_{(u_j+v_{j-1})t}(E) B_{(u_j+v_j)t}(E)\bigg\},
    \end{align}
    where $t > 0$.
    The assertion~\eqref{eq:the-asymptotics} follows from
    \begin{equation}
      - \lim_{L\to\infty} \frac{1}{\ln(L^{-a})}
        \int_0^\infty\dt\, t^{-1} e^{-tL^{-a}} f(t) =
        \lim_{t\to\infty} f(t),
    \end{equation}
    which is a consequence of Lemma~\ref{lemma:final-value-log} and of
    \begin{equation}
      \sup_{L > 1} \int_0^1\dt\, t^{-1} e^{-t L^{-a}} f(t) < \infty.
      \qedhere
    \end{equation}
\end{proof}

\begin{corollary} \label{corollary:the-asymptotics}
  For a.e.\ $E\in\R$, the estimate
  \begin{align}
    \tr\bigbraces{\bigparens{\1_{(-\infty,E]}(H_L)\1_{(E,\infty)}(H_L')}^n}
    \ge
    \ln L\; I_{2n}\eta_{2n}(E) + \oh(\ln L)
  \end{align}
  holds, where the $\oh(\ln L)$-term depends on $n$ and $E$.
  \begin{proof}
    We deduce from Theorem~\ref{theorem:the-asymptotics},
    Corollary~\ref{corollary:trace_bound_feynman_exponential} and from
    the arbitrariness of $E_0$ that
    \begin{equation} \label{eq:liminf-nth-term}
      \liminf_{L\to\infty}
      \frac{\tr\bigbraces{\bigparens{\1_{(-\infty,E]}(H_L)\1_{(E,\infty)}(H_L')}^n}}{\ln
      L}
    \ge
    a \, I_{2n}\eta_{2n}(E)
    \end{equation}
    for arbitrary $a\in(0,1)$ and
    a.e.\ $E\in\R$. Thus~\eqref{eq:liminf-nth-term} holds for $a = 1$
    and a.e.\ $E\in\R$. By definition of the limit inferior, this
    implies the claim.
  \end{proof}
\end{corollary}

\subsection{A multi-dimensional integral related to the Hilbert matrix}
\label{ssec:u-integral}

In this subsection, we compute the coefficient of $\eta_{2n}(E)$ in the
asymptotics in Corollary~\ref{corollary:the-asymptotics}, i.e., we compute
the integral
\begin{equation}
  I_n
  =
  \int_{(0,\infty)^n}\du\,
  \frac{\abs{u}_1e^{-\abs{u}_1}}{\prod_{j=1}^n (u_j + u_{j+1})}
\end{equation}
in Definition~\ref{definition:u-integral-finite}~(i).
Here, we use the convention $u_{n+1} = u_1$ for $u\in\R^n$.
We prove
\begin{theorem} \label{theorem:u-integral}
  Let $n\in\N_{\ge 2}$. Then
  \begin{equation} \label{eq:u-integral}
    I_n
    =
    (2\pi)^{n-2}
    \frac{\bigparens{\Gamma(\frac{n}{2})}^2}{\Gamma(n)}.
  \end{equation}
  This implies
  \begin{equation}
    I_{2n}
    =
    (2\pi)^{2n-2}
    \frac{[(n-1)!]^2}{(2n-1)!}
    =
    n J_{2n}
  \end{equation}
  for $n\in\N$, where $J_{2n}$ was defined in~\eqref{eq:J2n-def}.
\end{theorem}

We begin with an elementary lemma.
\begin{lemma} \label{eq:u-integral-diening-trick}
  Let $n\in\N_{\ge 2}$. Then
  \begin{equation} \label{eq:In-integral-formula}
    I_n = \frac{n}{2}\int_{(0,\infty)^n}\du\,
           \frac{e^{-\abs{u}_1}}{\prod_{j=1}^{n-1}(u_j + u_{j+1})}.
  \end{equation}
  \begin{proof}
    Using the
    symmetry of $I_n$ in the components of $u$, we compute

    \begin{align}
      I_n &
      = \frac{1}{2}\int_{(0,\infty)^n}\du\, e^{-\abs{u}_1}
      \frac{2\abs{u}_1}{\prod_{j=1}^n (u_j + u_{j+1})} \notag\\
      & =
      \frac{1}{2}\sum_{k=1}^n \int_{(0,\infty)^n}\du\, e^{-\abs{u}_1}
      \frac{u_k+u_{k+1}}{\prod_{j=1}^n (u_j + u_{j+1})} \notag\\
      & = \frac{n}{2} \int_{(0,\infty)^n}\du\,
      \frac{e^{-\abs{u}_1}}{\prod_{j=1}^{n-1}(u_j + u_{j+1})}.
      \qedhere
    \end{align}
  \end{proof}
\end{lemma}

In the sequel, we will work with the Rosenblum--Rovnyak integral operator
$T\from L^2((0,\infty))\to L^2((0,\infty))$,
see \cite{MR0094626} and \cite{MR0291899}, defined by
\begin{equation} \label{eq:rosenblum-rovnyak}
  (Tf)(x) := \int_0^\infty \dy\, \frac{e^{-(x+y)/2}}{x+y} f(y)
\end{equation}
for $f\in L^2((0,\infty))$ and $x\in(0,\infty)$.
This operator can be explicitly diagonalised. Following
\cite[Sec.~4.2]{MR2768565}, we define the Kontorovich--Lebedev
transform, i.e. the unitary operator
$U\from L^2((0,\infty))\to L^2((0,\infty))$,
\begin{equation} \label{eq:hilbert-matrix-unitary-operator}
  (Uf)(k) := \pi^{-1}\sqrt{k\sinh(2\pi k)} \, \abs{\Gamma(1/2 - ik)}
            \int_0^\infty\dx\, x^{-1} W_{0,ik}(x)f(x)
\end{equation}
for $f\in L^2((0,\infty))$ and $k\in(0,\infty)$, where $W_{0,ik}$
denotes the Whittaker function, see
\cite[Sec.~13.14]{DLMF} or \cite[Sec.~9.22--9.23]{MR2360010}. Then, the
spectral representation due to Rosenblum reads
\begin{equation} \label{eq:Hilbert-matrix-unitary-relation}
  (UTf)(k) = \frac{\pi}{\cosh(k\pi)}(Uf)(k)
\end{equation}
for $f\in L^2((0,\infty))$ and $k\in(0,\infty)$, see
\cite[Prop. 4.1]{MR2768565}.

\begin{proof}[Proof of Theorem~\ref{theorem:u-integral}]
  Let $n\in\N_{\ge 2}$. From
  \eqref{eq:In-integral-formula} and \eqref{eq:rosenblum-rovnyak}, we see that
  \begin{equation} \label{eq:Jn-Hilbert-matrix}
    \frac{2}{n} I_n
    =
    \bigangles{\phi_0, T^{n-1}\phi_0}_{L^2((0,\infty))}
  \end{equation}
  with $\phi_0(x) := e^{-x/2}$.
  From
  \eqref{eq:Hilbert-matrix-unitary-relation} and \eqref{eq:Jn-Hilbert-matrix}, we obtain
  \begin{equation} \label{eq:Jn-Hilbert-identity}
    \frac{2}{n} I_n
    =
    \bigangles{U\phi_0, UT^{n-1}\phi_0}_{L^2((0,\infty))}
    =
    \int_0^\infty\d k\, \abs{(U\phi_0)(k)}^2 \Bigparens{\frac{\pi}{\cosh(k\pi)}}^{n-1}.
  \end{equation}
  In order to compute $U\phi_0$, we employ the classical formula
  \begin{equation} \label{eq:classical-gamma}
    \abs{\Gamma(1/2 - ik)}^2 = \frac{\pi}{\cosh(k\pi)}
  \end{equation}
  for $k\in\R$, which is a consequence of the reflection formula for
  the Gamma function, and
  \begin{equation} \label{eq:classical-whittaker}
    \int_0^\infty\dx\, x^{-1} W_{0,ik}(x)e^{-x/2} =
    \frac{\pi}{\cosh(k\pi)}
    \where{k > 0},
  \end{equation}
  which follows from the special case $z=1/2$ and $\nu = \kappa = 0$ in
  \cite[Eq.~13.23.4]{DLMF}. From~\eqref{eq:hilbert-matrix-unitary-operator},
  \eqref{eq:classical-gamma}, and \eqref{eq:classical-whittaker}, we
  deduce
  \begin{equation}
    \abs{(U\phi_0)(k)}^2 = 2\pi k\frac{\sinh(k\pi)}{(\cosh(k\pi))^2}
  \end{equation}
  for $k > 0$.
  Inserting this into~\eqref{eq:Jn-Hilbert-identity} yields
  \begin{equation}
    \frac{2}{n} I_n
    = 2\pi^{n-2}\int_0^\infty\d k\, k
    \frac{\sinh k}{(\cosh k)^{n+1}}
    =
    \frac{2\pi^{n-2}}{n}\int_0^\infty\d k\,\frac{1}{(\cosh k)^n},
  \end{equation}
  where we applied the substitution $k\rightsquigarrow k/\pi$ and
  integrated by parts. This integral can be evaluated using the
  substitution $x = (\cosh k)^{-2}$:
  \begin{align}
    \frac{2}{n} I_n
    & =
    \frac{\pi^{n-2}}{n}
    \int_0^1\dx\, x^{n/2-1}(1-x)^{-1/2} 
     =
    \frac{\pi^{n-2}}{n} \Beta(n/2, 1/2),
  \end{align}
  where $\Beta$ denotes Euler's Beta Function. The claim follows
  from \cite[Eq.~8.384\,4 and 8.384\,1]{MR2360010}.
\end{proof}

\begin{remark}
  The Rosenblum--Rovnyak operator is the special case $T =
  \mathcal{H}_0$ in \cite[Eq.~(2.3)]{MR0094626} and is unitarily
  equivalent to the Hilbert matrix $\mathsf{H}\from\ell^2(\N_0)\to\ell^2(\N_0)$,
  \begin{equation}
    (\mathsf{H} c)_j = \sum_{k=0}^\infty \frac{c_k}{j+k+1}
  \end{equation}
  for $j\in\N_0$ and $c\in\ell^2(\N_0)$. In analogy
  to~\eqref{eq:Jn-Hilbert-matrix}, the
  representation
  \begin{equation}
    I_n = \frac{n}{2}\bigangles{e^{(0)}, \mathsf{H}^{n-1} e^{(0)}}_{\ell^2(\N_0)}
  \end{equation}
  holds with $e^{(0)} := (1,0,\dotsc)\in\ell^2(\N_0)$.
\end{remark}

\subsection{Relations to scattering theory}
\label{ssec:scattering-theory}

In order to complete the proof of Theorem~\ref{theorem:nth-term} we
need to relate the coefficient $\eta_{2n}(E)$ in
Definition~\ref{definition:density-diagonal} to the transition matrix
from scattering theory.
We begin with a definition.
\begin{definition}
  Let $\H_\ac(H)$ be the absolutely continuous subspace of the
  self-adjoint operator $H$. Then $\H_\ac(H)$ can be decomposed into a
  direct integral
  \begin{equation}
    \int_{\spec_\ac(H)}^{\oplus} \d E\,\H_E
  \end{equation}
  where $\H_E$ is a Hilbert space for every $E\in\spec_\ac(H)$. The operator $H$ acts on $\H_E$ by multiplication with the identity, see~\cite[\sectionsymbol
  1.5]{MR1180965}. This means that a vector $f\in\H_\ac(H)$
  corresponds to a vector-valued function $E\mapsto f_E\in\H_E$, and
  $Hf$ corresponds to $E\mapsto E f_E$.
\end{definition}

The transition matrix $T_E$ acts as a bounded operator on $\H_E$.
Moreover, we have the following representation.
\begin{lemma} \label{lemma:t-matrix-birman-Phi}
  The limit
  \begin{equation} \label{eq:birman-Phi}
    \Phi_\pm(E)
    := \lim_{\varepsilon\downto 0} \bigparens{I + \sqrt{V}(E\pm i\varepsilon - H')^{-1}\sqrt{V}}
  \end{equation}
  exists in the sense of convergence in operator norm for a.e.\ $E\in\R$. Moreover, there exists an
  operator $U(E)\from\H_\ac(H)\to\H_E$ such that
  $U(E)^*\,U(E)$ is the identity on $\ran\sqrt{A(E)}$ and the
  transition matrix $T_E\from\H_E\to\H_E$ satisfies
  \begin{equation}
    T_E
    =
    -2\pi i U(E) \tilde{T}(E) U(E)^*,
  \end{equation}
  where
  \begin{equation}
    \tilde{T}(E)
    :=
    -2\pi i\sqrt{A(E)}\Phi_+(E)\sqrt{A(E)}.
  \end{equation}
  \begin{proof}
    This is a result in abstract scattering theory, see
    e.g. \cite[\sectionsymbol 7]{BiEn67e}, \cite[\sectionsymbol 5.5]{MR1180965} or
    \cite[p.~394]{MR699113}. A detailed proof is given in \cite{Heiner}.
  \end{proof}
\end{lemma}

\begin{corollary} \label{corollary:t-matrix-birman-ops}
  The identity
  \begin{equation} \label{eq:t-matrix-birman-ops}
    \tilde{T}(E)^*\,\tilde{T}(E)
    =
    (2\pi)^2\sqrt{A(E)}B(E)\sqrt{A(E)}
  \end{equation}
  holds for a.e.\ $E\in\R$. In particular,
  \begin{equation} \label{eq:t-matrix-schatten-norm}
    \norm{T_E}_{2n}^{2n}
    =
    (2\pi)^{2n}\tr\bigbraces{\bigparens{A(E)B(E)}^n}
    =
    (2\pi)^{2n}\eta_{2n}(E)
  \end{equation}
  for every $n\in\N$,
  where $\norm{T_E}_{2n} := \sqrt[2n]{\tr\abs{T_E}^{2n}}$
  is the $2n$-Schatten norm of $T_E$.
\end{corollary}

\begin{proof}
    The operators $A(E)$ and $B(E)$ can be expressed as
    the operator limits
    \begin{align}
      - \pi A(E)
      & =
      \lim_{\varepsilon\downto 0}
        \Im\Bigparens{
          \sqrt{V}(E + i\varepsilon - H)^{-1}\sqrt{V}
        } \\
      - \pi B(E)
      & =
      \lim_{\varepsilon\downto 0}
        \Im\Bigparens{
          \sqrt{V}(E + i\varepsilon - H')^{-1}\sqrt{V}
        }
    \end{align}
    which exist for a.e.\ $E\in\R$, see \cite[Lemma~4.5]{BiEn67e}.
    From this and the second resolvent identity $(z - H')^{-1} - (z - H)^{-1}
    = (z - H')^{-1}V(z - H)^{-1}$ for $z\in\C\setminus\R$, the statement
    \begin{align} \label{eq:BviaA}
      & \Phi_+(E)^*A(E)\Phi_+(E) \notag \\ 
      & = -\frac{1}{\pi}
      	\lim_{\varepsilon\downto 0}
      	\Big\{\bigparens{I + \sqrt{V}(E-i\varepsilon - H')^{-1}\sqrt{V}} 
      	\Bigparens{\Im\bigparens{\sqrt{V}(E + i\varepsilon - H)^{-1}\sqrt{V} }} \notag\\
      & \hspace*{2cm} \times      
      	\bigparens{I + \sqrt{V}(E+i\varepsilon - H')^{-1}\sqrt{V}} \Big\} \notag \\
      & = \frac{1}{2\pi i} \lim_{\varepsilon\downto 0}
      	\Big\{\sqrt{V}\bigparens{I + (E-i\varepsilon-H')^{-1}V} 
            \big((E-i\varepsilon-H)^{-1} \notag\\
      & \hspace*{4cm}  - (E+i\varepsilon-H)^{-1} \big)  \bigparens{I + V(E+i\varepsilon-H')^{-1}}\sqrt{V} \Big\} \notag\\
      & =  \frac{1}{2\pi i} \lim_{\varepsilon\downto 0}
      \Big\{\sqrt{V} \bigparens{(E-i\varepsilon - H')^{-1} - (E+i\varepsilon - H')^{-1}} \sqrt{V} \Big\} \notag \\ 
      &  = B(E) 
    \end{align}
    follows and yields~\eqref{eq:t-matrix-birman-ops}. The unitary
    equivalence on $\ran\sqrt{A(E)}$ in
    Lemma \ref{lemma:t-matrix-birman-Phi} then implies~\eqref{eq:t-matrix-schatten-norm}.
\end{proof}

Corollary~\ref{corollary:t-matrix-birman-ops} yields the following
theorem.
\begin{theorem}\label{thm:scattering}
  Let $n\in\N$.
  For a.e.\ $E\in\R$
  \begin{equation} \label{eq:eta-as-t-matrix}
    \eta_{2n}(E)
    =
    \tr(\abs{T_E/(2\pi)}^{2n}),
  \end{equation}
  where $T_E:\H_E\to\H_E$ is the transition matrix for the
  energy~$E$.
\end{theorem}

%
%
\appendix
\section{Positivity of the exponent}
\label{ssec:positivity}

Here we consider the special case $V_0 = 0$ and show that the
decay exponent $\gamma(E)$ in \eqref{eq:gamma} is strictly
positive for a.e.\ $E > 0$. Throughout this appendix, we assume that
$V\ne 0$ satisfies~\eqref{eq:assumption:V_V0}.

\begin{theorem} \label{theorem:A-fourier-calculation}
  Let $V_0 = 0$. Let $E > 0$. Then the operator $A(E)$ from \eqref{eq:BirmanA} has
  the integral kernel
  \begin{equation} \label{eq:Birman-integral-kernel}
    A(E; x,y)
    =
    \frac{E^{d/2-1}}{2(2\pi)^d}
    \sqrt{V(x)}\sqrt{V(y)}\,
    \int_{\Sphere^{d-1}}\dS(\xi)\, e^{i\sqrt{E}\xi\cdot (x-y)}
  \end{equation}
  for a.e.\ $x,y\in\R^d$.
  Here, $\dS$ stands for integration with respect to the surface
  measure on the unit sphere $\Sphere^{d-1}\subseteq\R^d$.
  \begin{proof}
    Let $\varepsilon > 0$ and $f\in L^2(\R^d)$. Then, using the
    Fourier transform and spherical coordinates, we compute for a.e.~$x\in\R^d$
    \begin{align} \label{eq:fourier-computation}
      \Bigl( &\sqrt{V} \1_{(E-\varepsilon,E+\varepsilon)}(-\Delta)\sqrt{V}f\Bigr)(x)
      \notag\\
      & =
      \frac{\sqrt{V(x)}}{(2\pi)^{d}}\int_{\R^d}\d k\,\int_{\R^d}\dy\,
      \1_{(E-\varepsilon,E+\varepsilon)}(\abs{k}^2) e^{ik\cdot
      (x-y)} \sqrt{V(y)}f(y)
      \notag\\
      & =
      \frac{\sqrt{V(x)}}{2(2\pi)^{d}}\int_{\R^d}\dy\, \sqrt{V(y)}f(y)
      \int_{E-\varepsilon}^{E+\varepsilon}\dr\, r^{d/2-1} \int_{\Sphere^{d-1}}\dS(\xi)\,
      e^{i\sqrt{r}\xi\cdot (x-y)},
    \end{align}
    and therefore
    \begin{multline}
      \lim_{\varepsilon\downto 0}
      \frac{1}{2\varepsilon}
      \Bigparens{\sqrt{V} \1_{(E-\varepsilon,E+\varepsilon)}(-\Delta)\sqrt{V}f}(x)
      \\ =
      \frac{\sqrt{V(x)}}{2(2\pi)^{d}}\int_{\R^d}\dy\, \sqrt{V(y)}f(y)
      E^{d/2-1} \int_{\Sphere^{d-1}}\dS(\xi)\, e^{i\sqrt{E}\xi\cdot (x-y)},
    \end{multline}
    because the integrand in~\eqref{eq:fourier-computation} is
    continuous in $r$. This implies~\eqref{eq:Birman-integral-kernel}.
  \end{proof}
\end{theorem}

\begin{corollary} \label{corollary:free-A-infinte-range}
  Let $d\ge 2$. Let $V_0 = 0$ and $V\ne 0$. Then for any $E > 0$ the operator $A(E)$ from
  \eqref{eq:BirmanA} has infinite rank.
  \begin{proof}
    We first show that the set of functions
    \begin{equation} \label{eq:set-of-li-functions}
      \bigset{\R^d\ni x\mapsto \sqrt{V(x)}e^{i\xi\cdot x}\st
      \xi\in\R^d}
    \end{equation}
    is linearly independent. For this, notice that
    $\set{\C\ni z\mapsto e^{isz}\st s\in\R}$
    is linearly independent, since for $z = -ix$, these functions have
    different asymptotic behaviour for $x\to\infty$. Given a finite non-empty set
    $M\subseteq\R$ and $c_s\ne 0$ for $s\in M$, the analytic function
    $\C\ni z\mapsto \sum_{s\in M} c_s e^{isz}$ is therefore not identically
    zero, and thus $\R\ni x\mapsto \sum_{s\in M} c_s e^{isx}$ is zero
    only on a discrete subset of $\R$.

    Given another finite non-empty set $M\subseteq\R^d$ and $c_\xi\ne 0$ for
    $\xi\in M$, define $F\from\R^d\to\C$ via $F(x) := \sum_{\xi\in
    M}c_\xi e^{i\xi\cdot x}$. We show that
    $F^{-1}(\set{0})\subseteq\R^d$ is a null set. Since $F$ is
    continuous, this preimage is measurable with measure
    \[
      \int_{\R^d}\dx\, \1_{F^{-1}(\set{0})}(x)
      =
      \int_{\Sphere^{d-1}}\dS(\eta)\int_0^\infty\dr\,r^{d-1}\,
      \1_{\set{0}}\bigparens{F(r\eta)}
      = 0,
    \]
    where the $r$-integral is zero since for $\eta\in \Sphere^{d-1}$ fixed
    the function $r\mapsto F(r\eta) = \sum_{\xi\in M} c_\xi
    e^{ir\xi\cdot\eta}$ is zero only on a discrete subset of $\R$, as
    shown above. To show that the set~\eqref{eq:set-of-li-functions}
    is linearly independent, it suffices to show that
    \[
      \bigset{x\in\R^d\st \sqrt{V(x)}F(x) \ne 0}
      =
      \bigset{x\in\R^d\st V(x) \ne 0}
      \cap
      \bigset{x\in\R^d\st F(x) \ne 0}
    \]
    has positive measure. This is the case, since the first set in the
    intersection has positive measure and the second set is the
    complement of the null set $F^{-1}(\set{0})$.

    Now, let $f\in\ker A(E)$. Then
    \begin{equation}
      0 = \<f, A(E)f\> =
      \frac{E^{d/2-1}}{2(2\pi)^d}
      \int_{\Sphere^{d-1}}\dS(\xi)
      \Bigabs{\int_{\R^d}\dx \sqrt{V(x)} e^{i\sqrt{E}\xi\cdot x}
      f(x)}^2,
    \end{equation}
    and therefore
    \begin{equation} \label{eq:Vexp-f-orthogonal}
      \int_{\R^d}\dx \sqrt{V(x)} e^{i\sqrt{E}\xi\cdot x} f(x) = 0
    \end{equation}
    for a.e.\ $\xi\in \Sphere^{d-1}$. Since the left-hand
    side of~\eqref{eq:Vexp-f-orthogonal} is continuous in
    $\xi$,~\eqref{eq:Vexp-f-orthogonal} holds in fact for all $\xi\in
    \Sphere^{d-1}$. Since $f\in\ker A(E)$ was arbitrary, we conclude that
    \begin{equation}
      \bigset{\R^d\ni x\mapsto \sqrt{V(x)}e^{i\sqrt{E}\xi\cdot x}\st
      \xi\in \Sphere^{d-1}}
      \subseteq (\ker A(E))^{\perp}.
    \end{equation}
    Since $\Sphere^{d-1}$ is an infinite set for $d\ge 2$, the set of
    functions on the left-hand side is infinite and linearly
    independent, and thus $\dim(\ker A(E))^\perp =
    \infty$. Since the coimage $(\ker A(E))^\perp$ of the linear map
    $A(E)$ is isomorphic to $\ran A(E)$ (the restriction
    $A(E)\rvert_{(\ker A(E))^\perp}\from (\ker A(E))^\perp\to\ran
    A(E)$ being bijective), this shows $\dim\ran A(E) = \infty$.
  \end{proof}
\end{corollary}

\begin{remark}
  We expect Corollary~\ref{corollary:free-A-infinte-range} to
  generalise to the situation of non-zero background potentials $V_0$
  with suitable decay by using generalised eigenfunctions due to
  Ikebe-Povzner (see \cite[\sectionsymbol C5]{MR670130} and references
  therein) in place of $e^{i\sqrt{E}\xi\cdot x}$.
\end{remark}

The infinite rank of $A(E)$ implies positivity of
$\gamma(E)$.
\begin{theorem} \label{thm:T-matrix-infinite-rank} 
  Let $d\ge 2$ and $V_0 = 0$. Then the transition matrix 
  $T_E$ corresponding to the pair $H = -\Delta$ and $H' = -\Delta + V$ has infinite rank for a.e.\ $E > 0$. In particular, $T_E$ is non-zero
  and therefore
  \begin{equation}
    \gamma(E) =
    \pi^{-2}\bignorm{\arcsin\abs{T_E/2}}_{\HS}^2
    > 0
  \end{equation}
  for a.e.\ $E > 0$.
  \begin{proof}
    By Lemma~\ref{lemma:t-matrix-birman-Phi}, it suffices to show
    that $\tilde{T}(E) = -2\pi i\sqrt{A(E)} \Phi_+(E) $ $\sqrt{A(E)}$ has
    infinite rank, where $\Phi_\pm(E) = \lim_{\varepsilon\downto 0}
    \bigparens{I + \sqrt{V}(E\pm i\varepsilon-H')^{-1}\sqrt{V}}$. We
    show that its imaginary part $\Im \tilde{T}(E) =
    \frac{1}{2i}\bigparens{\tilde{T}(E) - \tilde{T}(E)^*}$ has infinite
    rank. For brevity, set $R := \lim_{\varepsilon\downto
    0}\sqrt{V}(E+i\varepsilon-H')^{-1}\sqrt{V}$. Recall that by the limiting
    absorption principle, this limit exists in operator norm for
    a.e.\ $E > 0$; in particular, $R$ is compact. We fix such an $E >
    0$ from now on.
    Then
    \begin{align}
      \Im\tilde{T}(E) & = \frac{1}{2i}\Bigl({}
      - 2\pi i\sqrt{A(E)}\Phi_+(E)\sqrt{A(E)} \notag\\
      & \hphantom{= \frac{1}{2i}\Bigl({}}\;
      - 2\pi i\sqrt{A(E)}\Phi_-(E)\sqrt{A(E)}\Bigr) \notag\\
      & =
      -2\pi\sqrt{A(E)}(I + \Re R)\sqrt{A(E)}.
    \end{align}
    Since $\Re R$ is compact, we can write it as $\Re R = R_1 + R_2$
    where $\norm{R_1} < 1/2$ and $R_2$ has finite rank. Thus
    \begin{align}
      -\frac{1}{2\pi}\Im\tilde{T}(E)
      =
      \sqrt{A(E)}( I + R_1 )\sqrt{A(E)} + \tilde{A}
    \end{align}
    where $\tilde{A}$ is a finite rank operator. Now, since $I + R_1
    \ge I - \frac{1}{2}I = \frac{1}{2}I$, we get
    \begin{equation}
      \sqrt{A(E)}(I + R_1)\sqrt{A(E)}
      \ge
      \tfrac{1}{2}A(E).
    \end{equation}
    By Corollary~\ref{corollary:free-A-infinte-range}, this operator
    has infinite rank.
  \end{proof}
\end{theorem}

\section*{Acknowledgements}
We are grateful to Alexander Pushnitski for communicating the
argument in Theorem~\ref{thm:T-matrix-infinite-rank} to us and for
further interesting discussions. We also thank Lars Diening and Parth
Soneji for helpful discussions.



\enlargethispage{.4cm}

\end{document}